\documentclass[11pt,a4paper]{article}
\pdfoutput=1

\usepackage[margin=1in]{geometry}
\usepackage{amsmath,amssymb,amsthm,booktabs,color,doi,enumitem,graphicx,latexsym,url,xcolor}
\usepackage[numbers,sort&compress]{natbib}
\usepackage{microtype,hyperref}
\usepackage{thmtools,thm-restate}
\usepackage[symbol]{footmisc}

\definecolor{citecolor}{HTML}{0000C0}
\definecolor{urlcolor}{HTML}{000080}

\hypersetup{
    colorlinks=true,
    linkcolor=black,
    citecolor=citecolor,
    filecolor=black,
    urlcolor=urlcolor,
}

\newtheorem{theorem}{Theorem}
\newtheorem{lemma}[theorem]{Lemma}
\newtheorem{corollary}[theorem]{Corollary}

\theoremstyle{remark}
\newtheorem{remark}[theorem]{Remark}

\newcommand{\namedref}[2]{\hyperref[#2]{#1~\ref*{#2}}}
\newcommand{\sectionref}[1]{\namedref{Section}{#1}}

\newcommand{\theoremref}[1]{\namedref{Theorem}{#1}}
\newcommand{\corollaryref}[1]{\namedref{Corollary}{#1}}
\newcommand{\figureref}[1]{\namedref{Figure}{#1}}
\newcommand{\lemmaref}[1]{\namedref{Lemma}{#1}}
\newcommand{\tableref}[1]{\namedref{Table}{#1}}
\newcommand{\remarkref}[1]{\namedref{Remark}{#1}}

\renewcommand{\vec}[1]{\mathbf{#1}}

\newcommand{\totaln}{\eta}

\DeclareMathOperator{\polylog}{polylog}
\DeclareMathOperator{\E}{\mathbf E}

\newenvironment{mycover}
               {\list{}{\listparindent 0pt
                        \itemindent    \listparindent
                        \leftmargin    0pt
                        \rightmargin   0pt
                        \parsep        0pt}%
                \raggedright
                \item\relax}
               {\endlist}

\begin{document}

\hypersetup{
    pdfauthor={Christoph Lenzen, Joel Rybicki, Jukka Suomela},
    pdftitle={Efficient Counting with Optimal Resilience},
}

\vspace*{2cm}
\begin{mycover}
{\LARGE \textbf{Efficient Counting with Optimal Resilience\footnote[1]{The manuscript is an extended and revised version of two preliminary conference reports~\cite{lenzen15towards,lenzen15efficient} that appeared in the Proceedings of the 34th Annual ACM Symposium on Principles of Distributed Computing (PODC 2015) and in the Proceedings of the 29th International Symposium on Distributed Computing (DISC 2015).}}\par}

\bigskip
\bigskip

\medskip
\textbf{Christoph Lenzen}

\smallskip
{\small Department of Algorithms and Complexity, \\
Max Planck Institute for Informatics\par}

\bigskip
\textbf{Joel Rybicki}

\smallskip
{\small Helsinki Institute for Information Technology HIIT, \\
Department of Computer Science, Aalto University\par}

\medskip
{\small Department of Algorithms and Complexity, \\
Max Planck Institute for Informatics\par}

\bigskip
\textbf{Jukka Suomela}

\smallskip
{\small Helsinki Institute for Information Technology HIIT, \\
Department of Computer Science, Aalto University\par}

\end{mycover}

\bigskip
\paragraph{Abstract.}

Consider a complete communication network of $n$ nodes, where the nodes receive a common clock pulse. We study the synchronous $c$-counting problem: given any starting state and up to $f$ faulty nodes with arbitrary behaviour, the task is to eventually have all correct nodes labeling the pulses with increasing values modulo $c$ in agreement. Thus, we are considering algorithms that are self-stabilising despite Byzantine failures. In this work, we give new algorithms for the synchronous counting problem that (1)~are deterministic, (2)~have optimal resilience, (3)~have a linear stabilisation time in $f$ (asymptotically optimal), (4)~use a small number of states, and consequently, (5) communicate a small number of bits per round. Prior algorithms either resort to randomisation, use a large number of states and need high communication bandwidth, or have suboptimal resilience. In particular, we achieve an \emph{exponential} improvement in both state complexity and message size for deterministic algorithms. Moreover, we present two complementary approaches for reducing the number of bits communicated during and after stabilisation.

\thispagestyle{empty}
\setcounter{page}{0}
\newpage

\section{Introduction}\label{sec:intro}

In this work, we design space- and communication-efficient, self-stabilising, Byz\-antine
fault-tolerant algorithms for the \emph{synchronous counting problem}. We are
given a complete communication network on $n$ nodes, with arbitrary initial
states. There are up to $f$ faulty nodes. The task is to synchronise the nodes
so that all non-faulty nodes will count rounds modulo~$c$ in agreement. For
example, here is a possible execution for $n = 4$ nodes, $f = 1$ faulty node,
and counting modulo $c = 3$; the execution stabilises after $t = 5$ rounds:
\begin{center}
    \begin{tabular}{@{}l@{\ \quad}lllllllllll@{\ \ }ll@{}}
    & \multicolumn{5}{@{}l}{Stabilisation}
    & \multicolumn{7}{l}{Counting} \\
    \cmidrule[\heavyrulewidth](r){2-6}
    \cmidrule[\heavyrulewidth](lr){7-13}
    Node 1: & 2 & 2 & 0 & 2 & 0 & 0 & 1 & 2 & 0 & 1 & 2 & \ldots\! \\
    Node 2: & 0 & 2 & 0 & 1 & 0 & 0 & 1 & 2 & 0 & 1 & 2 & \ldots\! \\
    Node 3: & \multicolumn{11}{@{}l}{\emph{faulty node, arbitrary behaviour}} & \ldots\! \\
    Node 4: & 0 & 0 & 2 & 0 & 2 & 0 & 1 & 2 & 0 & 1 & 2 & \ldots\! \\
    \end{tabular}
\end{center}

Synchronous counting is a coordination primitive that can be used e.g.\ in large
integrated circuits to synchronise subsystems to easily implement
\emph{mutual exclusion} and \emph{time division multiple access} in a
fault-tolerant manner. Note that in this context, it is natural to assume that a
synchronous clock signal is available, but the clocking system usually does not
provide explicit round numbers. Solving synchronous counting thus yields 
highly dependable round counters for subcircuits.

If we neglect communication, counting and consensus are essentially
equivalent~\cite{dolev07actions,dolev13counting,dolev15counting}. In particular,
many lower bounds on (binary) consensus directly apply to the counting
problem~\cite{dolev85bounds,pease80reaching,fischer82lower}. However, the known
generic reduction of counting to consensus incurs a factor-$f$ overhead in space
and message size. In this work, we present techniques that reduce the
number of bits nodes broadcast in each round to $O(\log^2 f + \log c)$.

\subsection{Contributions}

Our contributions constitute of two parts. First, we give novel space-efficient deterministic algorithms for synchronous counting with optimal resilience and fast stabilisation time. Second, we show
how to extend these algorithms in a way that reduces the number of communicated bits \emph{during} and \emph{after} stabilisation.

\paragraph{Space-efficient counting algorithms.} In this work, we take the following approach for devising \emph{communication-efficient} counting algorithms: we first design \emph{space-efficient} algorithms, that is, algorithms in which each node stores only a few bits between consecutive rounds. Space-efficient algorithms are particularly attractive from the perspective of fault-tolerant systems: if we can keep the number of state bits small, we can also reduce the overall complexity of the system, which in turn makes it easier to use highly reliable components for an implementation.

Once we have algorithms that only need a small number of bits to encode the local state of a node, we also get algorithms that use small messages: the nodes can simply broadcast their entire state to everyone. Our main result is summarised in the following theorem; here \emph{$f$-resilient} means that we can tolerate up to $f$ faulty nodes:

\begin{restatable}{theorem}{mainthm}\label{thm:main-theorem}
For any integers $c,n > 1$ and $f < n/3$, there exists a deterministic $f$-resilient
synchronous $c$-counter that runs on $n$ nodes, stabilises in $O(f)$ rounds, and uses $O(\log^2 f+\log c)$ bits to encode the state of a node. 
\end{restatable}

Our main technical contribution is a recursive construction that shows how to ``amplify''
the resilience of a synchronous counting algorithm. Given a synchronous counter
for some values of $n$ and $f$, we will show how to design synchronous counters
for larger values of $n$ and $f$, with a very small increase in time and state
complexity. This has two direct applications:
\begin{enumerate}
    \item From a practical perspective, we can apply existing computer-designed algorithms (e.g.\ $n = 4$ and $f = 1$) as a building block in order to design efficient deterministic algorithms for a moderate number of nodes (e.g., $n = 36$ and $f = 7$).
    \item From a theoretical perspective, we can design deterministic algorithms for synchronous counting for any $n$ and any $f < n/3$, with a stabilisation time of $\Theta(f)$, and with only $O(\log^2 f)$ bits of state per node.
\end{enumerate}
The state complexity and message size is an \emph{exponential} improvement over prior work, and
the stabilisation time is asymptotically optimal for deterministic algorithms~\cite{fischer82lower}. 

\paragraph{Reducing communication after stabilisation.} In our deterministic algorithms, each node only needs to store a few number of bits between consecutive rounds, and thus, a node can e.g. afford to broadcast its entire state to all other nodes in each round. Moreover, we present a technique to reduce the number of communicated bits further.

We give a deterministic construction in which \emph{after}
stabilisation each node broadcasts $O(1 + B\log B)$ bits every $\kappa$ rounds, where $B = O(\log c / \log \kappa)$, for an
essentially unconstrained choice of $\kappa$, at the expense of additively
increasing the stabilisation time by $O(\kappa)$. In particular, for the special
case of optimal resilience and polynomial counter size, we obtain the following
result.

\begin{restatable}{corollary}{SILENCE}\label{coro:SILENCE!}
For any $n > 1$ and $c = n^{O(1)}$ that is an integer multiple of $n$,
there exists a synchronous $c$-counter that runs on $n$ nodes, has
optimal resilience $f=\lfloor (n-1)/3\rfloor$, stabilises in $\Theta(n)$
rounds, requires $O(\log^2 n)$ bits to encode the state of a node, and
for which after stabilisation correct nodes broadcast aysmptotically optimal
$O(1)$ bits per $\Theta(n)$ rounds.
\end{restatable}
We remark that in the above result we simply reduce the frequency of communication and the size of messages instead of e.g.\ bounding the number of nodes communicating in any given round (known as broadcast efficiency)~\cite{takimoto12wireless}. In our work, we exploit synchrony after stabilisation to schedule communication, and thus, our approach is to be contrasted with attempting to reduce the total number of communication partners or communicating nodes after stabilisation~\cite{delporte-gallet07robust,devismes09communication,takimoto12wireless}.

\paragraph{Reducing the number of messages.}

To substantiate the conjecture that finding algorithms with small state complexity may lead to highly communication-efficient solutions, we proceed to consider a slightly stronger synchronous \emph{pulling model}. In this model, a node may send a request to another node and receive a response in a single round, based on the state of the responding node at the beginning of the round. The cost for the exchange is then attributed to the pulling node; in a circuit, this translates to each node being assigned an energy budget that it uses to ``pay'' for the communication it triggers. In this model, it is straightforward to combine our recursive construction used in \theoremref{thm:main-theorem} with random sampling to obtain the following results:
\begin{enumerate}
    \item We can achieve the same asymptotic running time and state complexity as the deterministic algorithm from \theoremref{thm:main-theorem} with each node pulling only $\polylog n$ messages in each round. The price is that the resulting algorithm retains a probability of $n^{-\polylog n}$ to fail in each round even after stabilisation and that the resilience is $f < n/(3+\gamma)$ for any constant $\gamma > 0$.
    \item If the failing nodes are chosen independently of the algorithm, we can fix the random choices. This results in a pseudorandom algorithm which stabilises with a probability of $1-n^{-\polylog n}$ and in this case keeps counting correctly.
\end{enumerate}

\subsection{Our Approach}

Most prior deterministic algorithms for synchronous counting and closely-related problems utilise consensus protocols~\cite{hoch06digital,dolev07actions}. Indeed, if we ignore space and communication, reductions exist both ways showing that the problems are more or less equivalent~\cite{dolev15survey}; see \sectionref{sec:related} for further discussion on prior work.

However, to construct fast space- and communication-efficient counters, we are facing a chicken-and-egg problem:
\begin{itemize}
    \item \textbf{From counters to consensus:} If the correct nodes could agree on a counter, they could jointly run a \emph{single} instance of synchronous consensus.
    \item \textbf{From consensus to counters:} If the nodes could run a consensus algorithm, they could agree on a counter.
\end{itemize}
A key step to circumvent this obstacle is the following observation:
\begin{itemize}
    \item \textbf{From unreliable counters to consensus:} If the correct nodes can agree on a counter \emph{at least for a while}, they can jointly run a single instance of consensus.
    \item \textbf{From consensus to reliable counters:} Consensus can be then used to facilitate agreement on the output counter, and it is possible to maintain agreement even if the underlying unreliable counters fail later on.
\end{itemize}
The task of constructing counters that are correct only once in a while is easier;
in particular, it does not require that we solve consensus in the process. As our main technical result,
we show how to ``amplify'' the resilience $f$, at a cost of losing some guarantees:
\begin{itemize}
    \item \textbf{Input:} Two counters with a small $f$; guaranteed to work permanently after stabilisation.
    \item \textbf{Output:} A counter with a large $f$; guaranteed to work only once in a while.
\end{itemize}
This can be then used to jointly run a single instance of consensus and stabilise the output. We show how to obtain such a counter based on simple local consistency checks, timeouts, and threshold voting.

In the end, a recursive application of this scheme allows us to build space-efficient counting algorithms for any $n$ with optimal resilience. At each level of recursion, we only need to run a single instance of consensus. As there will be $O(\log f)$ levels of recursion, in total each node participates in only $O(\log f)$ consensus instances.

\subsection{Structure}

\sectionref{sec:related} reviews prior work on impossibility results and counting algorithms. \sectionref{sec:preliminaries} provides a formal description of the basic model of computation and the synchronous counting problem. \sectionref{sec:construction} gives the main technical result on resilience boosting, and \sectionref{sec:recursion} applies it to construct fast and communication-efficient algorithms. \sectionref{sec:silence} shows how to reduce the number of bits communicated during and after stabilisation. \sectionref{sec:random} discusses the pulling model and randomised sampling.

\section{Related Work}\label{sec:related}

In this section, we first overview impossibility results related to counting, and then discuss both deterministic and randomised algorithms for the counting problem.

\paragraph{Impossibility results.}

As mentioned, counting is closely related to consensus as reductions exist both ways~\cite{dolev15survey}: consensus can be solved in time $O(T)$ tolerating $f$ faults if and only if counting can be solved in time $O(T)$ tolerating $f$ faults.

With this equivalence in mind, several impossibility results for consensus directly hold for counting as well. First, consensus cannot be solved in the presence of $n/3$ or more Byzantine failures~\cite{pease80reaching}. Second, any deterministic $f$-resilient consensus algorithm needs to run for at least $f+1$ communication rounds~\cite{fischer82lower}. Third, it is known that the connectivity of the communication network must be at least $2f+1$~\cite{dolev82byzantine}. Finally, any consensus algorithm needs to communicate at least $\Omega(nf)$ bits in total~\cite{dolev85bounds}. 

In terms of communication complexity, no better bound than $\Omega(nf)$ on the \emph{total} number of communicated bits is known. While non-trivial for consensus, this bound turns out to be trivial for deterministic counting algorithms: a self-stabilising algorithm needs to verify its output, and to do that, each of the $n$ nodes needs to receive information from at least $f+1 = \Omega(f)$ other nodes to be certain that some other non-faulty node has the same output value. Similarly, no non-constant lower bounds on the number of \emph{state bits} nodes are known; however, a non-trivial constant lower bound for the case $f=1$ is known~\cite{dolev15counting}. 

\paragraph{Prior algorithms.}

\begin{table}[t!]
\centering
\begin{tabular}{@{}l@{\quad\ \ }l@{\quad\ \ }l@{\quad\ \ }l@{\quad\ \ }l@{}}
  \toprule
  resilience & stabilisation time & state bits & deterministic & reference \\
  \midrule
  $f<n/3$ & $O(1)$ & $n^{O(1)}$ & no & \cite{ben-or08fast}\ $^{(*)}$ \\
  $f<n/3$ & $O(f)$ & $O(f \log f)$ & yes & \cite{dolev07actions} \\ 
  $f<n/3$ & $2^{2(n-f)}$ & 2 & no & \cite{dolev00self-stabilization,dolev04clock-synchronization} \\
  $f<n/3$ & $\min \bigl\{ 2^{2f+2} + 1,\, 2^{O(f^2/n)} \bigr\}$ & 1 & no & \cite{dolev15counting} \\
  $f=1$, $n\geq 4$ & 7 & 2 & yes & \cite{dolev15counting} \\
  $f = n^{1-o(1)}$ & $O(f)$ & $O(\log^2 f/\log \log f)$ & yes &
  \cite{lenzen15towards}\\
  \midrule
  $f<n/3$ & $O(f)$ & $O(\log^2 f)$ & yes & this work \\
  \bottomrule
\end{tabular}
\caption{Summary of counting algorithms for the case $c=2$. For randomised
algorithms, we list the expected stabilisation time. $^{(*)}$~The solution
from~\cite{ben-or08fast} relies on a shared coin---details vary, but all known shared coins with large resilience require large
states and messages.}\label{table:algorithms}
\end{table}

There are several algorithms to the synchronous counting problem, with different trade-offs in terms of resilience, stabilisation time, space complexity, communication complexity, and the use of random bits. For a brief summary, see \tableref{table:algorithms}.

Designing space-efficient \emph{randomised} algorithms for synchronous counting
is fairly
straightforward~\cite{dolev00self-stabilization,dolev04clock-synchronization,dolev15counting}:
for example, the nodes can simply choose random states until a clear majority of
nodes has the same state, after which they start to follow the majority.
Likewise, given a shared coin, one can quickly reach agreement by defaulting to
the coin whenever no clear majority is observed~\cite{ben-or08fast}. However,
existing highly-resilient shared coins are very inefficient in terms of communication or need additional assumptions, such as private communication links between correct nodes.
Less resilient shared coins are easier to obtain: resilience $\Theta(\sqrt{n})$ is achieved by each node announcing the outcome of an independent coin flip and locally outputting the (observed) majority value. In addition, $\Omega(n / \log^2 n)$-resilient Boolean functions give fast communication-efficient coins~\cite{ajtai93influence}.
Designing quickly stabilising algorithms that are both communication-efficient and
space-efficient has turned out to be a challenging
task~\cite{dolev07actions,dolev13counting,dolev15counting}, and it remains open
to what extent randomisation can help in designing such algorithms.

In the case of \emph{deterministic} algorithms, algorithm synthesis has been used for
computer-aided design of optimal algorithms with resilience $f=1$, but the
approach does not scale due to the extremely fast-growing space of possible
algorithms~\cite{dolev15counting}. In general, many fast-stabilising algorithms build on a
connection between Byzantine consensus and synchronous counting, but require a
large number of states per node \cite{dolev07actions} due to, e.g., running a
large number of consensus instances in parallel. 
Recently, in one of the preliminary conference reports~\cite{lenzen15towards} this paper is based on, we outlined a recursive approach where each node needs to participate in only $O(\log f / \log \log f)$ parallel instances of consensus. However, this approach resulted in suboptimal resilience of $f = n^{1-o(1)}$. 

Finally, we note that while counting algorithms are usually designed for the case of a fully-connected communication topology, the algorithms can be extended to use in a variety of other graph classes with high enough connectivity~\cite{dolev15counting}.

\paragraph{Related problems.} Boczkowski et al.~\cite{boczkowski17minimizing} study the synchronous $c$-counting problem (under the name self-stabilising clock synchronisation) with $O(\sqrt{n})$ Byzantine faults in a stochastic communication setting that resembles the pulling model we consider in \sectionref{sec:random}. However, their communication model is much more restricted: in every round, each node interacts with at most constantly many nodes which are chosen uniformly at random. Moreover, nodes only exchange messages of size $O(\log c)$ bits. 

Without Byzantine (or other types of permanent) faults, self-stabilising counters and digital clocks have been studied as the \emph{self-stabilising unison problem}~\cite{gouda90unison,arora91maintaining,boulinier08unison}. However, unlike in the fully-connected setting considered in this work, the underlying communication topology in the unison problem is typically assumed to be an arbitrary graph. In our model, in absence of permanent faults the problem becomes trivial, as nodes may simply reproduce the clock of a predetermined leader. The unison problem has also been studied in asynchronous models~\cite{boulinier08unison,dubois12unison}; this variant is also known as self-stabilising synchronisers~\cite{awerbuch07time-optimal}. 

\section{Preliminaries}\label{sec:preliminaries}

In this section, we define the model of computation and the counting problem.

\subsection{Model of Computation}

We consider a fully-connected synchronous message-passing network. That is, our
distributed system consists of a network of $n$ nodes, where each node is a
state machine and has communication links to all other nodes in the network. All
nodes have a unique identifier from the set $[n] = \{0, 1, \ldots, n-1\}$. The
computation proceeds in synchronous communication rounds. In each round, all
nodes perform the following in a lock-step fashion: 
\begin{enumerate}[noitemsep]
 \item \emph{broadcast} a single message to all nodes,
 \item \emph{receive} messages from all nodes, and
 \item \emph{update} the local state.
\end{enumerate}

We assume that the initial state of each node is arbitrary and there are up to
$f$ Byzantine nodes. A Byzantine node may have arbitrary behaviour, that
is, it can deviate from the protocol in any manner. In particular, the Byzantine
nodes can collude together in an adversarial manner and a single Byzantine
node can send \emph{different} messages to different correct nodes.

\paragraph{Algorithms.}

Formally, we define an algorithm as a tuple $\vec A = \langle X, g, p \rangle$,
where $X$ is the set of all states any node can have, $g \colon [n] \times X^n
\to X$ is the \emph{state transition function}, and $p \colon [n] \times X \to
[c]$ is the \emph{output function}. At each round when node $v$
receives a vector $\vec x = \langle x_0, \ldots, x_{n-1} \rangle \in X^n$ of messages,
node $v$ updates it state to $g(v, \vec x)$ and outputs $p(v, x_v)$. 
As we consider $c$-counting algorithms, the set of output values is the
set $[c] = \{0,1,\ldots,c-1\}$ of counter values. 

The tuples passed to the state transition function $g$ are
ordered according to the node identifiers. Put otherwise, the nodes can identify the sender of a message---this is frequently referred to as source authentication. Moreover, in the basic model, we assume that all nodes simply broadcast their state to all other nodes. Thus, the set of messages is the same as the set of possible states. 

\paragraph{Executions.}

For any set of $\mathcal{F} \subseteq [n]$ of faulty nodes, we define a projection
$\pi_\mathcal{F}$ that maps any state vector $\vec x \in X^n$ to a
\emph{configuration} $\pi_F(\vec x) = \vec e$, where $e_v = *$ if $v \in
\mathcal{F}$ and $e_v = x_v$ otherwise. That is, the values given by Byzantine
nodes are ignored and a configuration consists of only the states of correct
nodes. A configuration $\vec d$ is \emph{reachable} from configuration $\vec e$
if for every correct node $v \notin \mathcal{F}$ there exists some $\vec x \in
X^n$ satisfying $\pi_\mathcal{F}(\vec x) = \vec e$ and $g(v, \vec x) = d_v$. An \emph{execution} of an algorithm $\vec A$ is an infinite
sequence of configurations $\xi = \langle \vec e_0, \vec e_1 \ldots, \rangle$
where configuration $\vec e_{r+1}$ is reachable from configuration $\vec e_r$.

\subsection{Synchronous Counters and Complexity Measures}

We say that an execution $\xi = \langle \vec e_0, \vec e_1 \ldots, \rangle$ of a counting algorithm $\vec A$ \emph{stabilises} in time $T$ if
there is some $k \in [c]$ such that for every correct node $v \in [n] \setminus \mathcal{F}$ it holds that
\[
p(v, e_{T+r,v}) = r - k \bmod c \ \text{ for all } r \ge 0,
\]
where $e_{T+r,v} \in X$ is the state of node $v$ in round $T+r$. 

An algorithm $\vec A$ is said to be a \emph{synchronous $c$-counter} with
\emph{resilience} $f$ that stabilises in time $T$, if for every $\mathcal{F}
\subseteq [n]$, $|\mathcal{F}| \le f$, all executions of algorithm $\vec A$
stabilise within $T$ rounds. In this case, we say that the \emph{stabilisation
time} $T(\vec A)$ of $\vec A$ is the minimal such $T$ that all executions of $\vec A$ stabilise in $T$ rounds. The \emph{state
complexity} of $\vec A$ is $S(\vec A) = \lceil \log |X| \rceil$, that is, the
number of bits required to encode the state of a node between subsequent rounds. For brevity, we will
often refer to $\mathcal{A}(n,f,c)$ as the family of synchronous $c$-counters
over $n$ nodes with resilience $f$. For example, $\vec A \in \mathcal{A}(4, 1,
2)$ denotes a synchronous 2-counter (i.e.\ a binary counter) over 4 nodes tolerating one failure.

\section{Boosting Resilience}\label{sec:construction}

In this section, we show how to use existing ``small'' synchronous counters to construct new ``large'' synchronous counters with a higher resilience $f$ and a larger number of nodes $n$; we call this \emph{resilience boosting}. We will then apply the idea recursively, with trivial counters as a base case.

\subsection{Road Map}

The high-level idea of resilience boosting is as follows. We start with counters that have a low resilience $f'$ and use a small number of nodes $n'$. We use such counters to construct a new ``weak'' counter that has a higher resilience $f > f'$ and a large number of nodes $n > n'$ but only needs to behave correctly \emph{once in a while} for sufficiently long. Once such a weak counter exists, it can be used to provide consistent round numbers for long enough to execute a \emph{single} instance of a high-resilience consensus protocol. This can be used to reach agreement on the output counter.

\paragraph{Constructing the Weak Counter.}

For clarity, we will use here the term \emph{strong counter} to refer to a self-stabilising fault-tolerant counter in the usual sense, and the term \emph{weak counter} to refer to a counter that behaves correctly once in a while. We assume that $f'$-resilient strong counters for all $f'<f$ already exist, and we show how to construct an $f$-resilient weak counter that behaves correctly for at least $\tau$ rounds. Put slightly more formally, a weak $\tau$-counter satisfies the following property: there exists a round $r$ such that for all correct nodes $v,w \in V \setminus F$ satisfy
\begin{itemize}[noitemsep]
 \item $d(v,r) = d(w,r)$ and
 \item $d(v,r') = d(v,r'-1) +1 \mod \tau$ for all $r' \in \{ r+1, \ldots, r + \tau - 1 \}$,
\end{itemize}
where $d(v,r)$ denotes the value of the weak counter at node $v$ in round $r$. That is, eventually there will be $\tau$ consecutive rounds during which the (weak) counter values agree and are incremented by one modulo $\tau$ every round. However, after these $\tau$ rounds, the counters can behave arbitrarily.

Let $f_0+f_1+1=f$ and $n_0+n_1=n$. We take an $f_0$-resilient strong $2\tau$-counter ${\vec A}_0$ with $n_0$ nodes and an $f_1$-resilient strong $6\tau$-counter ${\vec A}_1$ with $n_1$ nodes, and use them to construct an $f$-resilient weak counter with $n$ nodes.

We partition $n$ nodes in disjoint ``blocks'': block 0 runs ${\vec A}_0$ with $n_0$ nodes and block 1 runs ${\vec A}_1$ with $n_1$ nodes.
At least one of the algorithms will eventually stabilise and count correctly.
The key challenge is making sure that eventually all correct nodes (in both blocks!) will follow the same correct counter, at least for $\tau$ rounds.

To this end, each block maintains a \emph{leader pointer}.
The leader pointers are changed regularly:
block $0$ changes its leader pointer every $\tau$ rounds, and
block $1$ changes its leader pointer every $3\tau$ rounds.
If the leader pointers behave correctly, there will be regularly periods of $\tau$ rounds such that both of the leader pointers point to the same correct block.

If we had reliable counters, block $i$ could simply use the current value of counter ${\vec A}_i$ to determine the current value of its leader pointer.
However, one of the counters might misbehave.
As a remedy, each node $v$ of block $i$ checks if the output variable of counter ${\vec A}_i$ increases by $1$ in each round. If not, it will consider ${\vec A}_i$ faulty for $\Theta(\tau)$ rounds. 
The final output of a node is determined as follows:
\begin{itemize}
    \item If node $v$ in block $i$ thinks that ${\vec A}_i$ is faulty, it outputs the current value of counter ${\vec A}_{1-i}$.
    \item Otherwise, it uses the current value of ${\vec A}_i$ to construct the leader pointer $\ell \in \{0,1\}$, and it outputs the current value of counter ${\vec A}_{\ell}$.
\end{itemize}
Note that the counter ${\vec A}_i$ might seem to be behaving in a faulty manner if there has not been enough time for $\vec A_i$ to stabilise. However, each node $v$ of block $i$ will consider a block to be faulty at most $\Theta(\tau)$ rounds before checking again whether the output of $\vec A_i$ behaves consistently. Thus, if $\vec A_i$ eventually stabilises, then eventually node $v$ stops considering $\vec A_i$ as faulty for good (at least until the next transient failure).

The above consistency check almost cuts it---except that two nodes $w\neq v$ of block $i$ may have
different opinions on the current value of ${\vec A}_i$. We clear this final hurdle
by enlisting the help of \emph{all} nodes for a majority vote on what
the current value of ${\vec A}_i$ actually is. Essentially, we use
threshold voting; this way all nodes that think that ${\vec A}_i$
behaves correctly will agree on a globally unique counter value $\alpha_i$ for ${\vec A}_i$.

If, for example, block $0$ contains at most $f_0$ faulty nodes, all of this eventually entails the following:
\begin{enumerate}
    \item Counter ${\vec A}_0$ stabilises, counts correctly, and all correct nodes agree on its counter value $\alpha_0$.
    \item All correct nodes of block $0$ think that block $0$ is counting correctly. They use $\alpha_0$ to derive the value of the leader pointer. Once in $2\tau$ rounds, when the $2\tau$-counter $\alpha_0$ wraps around to $0$, the pointer switches to $0$, and the nodes will output the counter value $\alpha_0$ for $\tau$ rounds.
    \item Some correct nodes of block $1$ may think that block $1$ is counting correctly for $\Theta(\tau)$ rounds. While this is the case, all of them agree on a value $\alpha_1$ that increases by $1$ in each round. This value is used to derive the leader pointer of block $1$. Once in $6\tau$ rounds, when the $6\tau$-counter $\alpha_1$ wraps around to $0$, the pointer will switch to $0$, and the nodes will output the value of $\alpha_0$ for $3\tau$ rounds (as the leader pointer does not change for $3\tau$ rounds).
    \item Some correct nodes of block $1$ may detect that block $1$ is faulty. Such nodes will output the value of $\alpha_0$ for $\Theta(\tau)$ rounds.
    \item In summary, eventually there will be $\tau$ consecutive rounds during which  all correct nodes output the same counter value $\alpha_0$.
\end{enumerate}
The other case (block $1$ has at most $f_1$ faulty nodes) is analogous.

\paragraph{Using the Weak Counter.}

Now we have constructed a counter that will eventually produce a consistent output for at least $\tau$ rounds.
We leverage this property to execute the \emph{phase king}
consensus protocol~\cite{berman89consensus} to stabilise the output counters.
The protocol will have the following crucial property: if all nodes agree on the output,
then even if the round counter becomes inconsistent, the agreement on the output persists.
Thus, it suffices for us that $\tau$ is large enough to enable the nodes to consistently execute
the phase king algorithm once to reach agreement; $\tau = O(f)$ will do.

The stabilisation time on each level is the maximum of
the stabilisation times of counters ${\vec A}_i$ plus $O(\tau) = O(f)$;
by choosing $f_1\approx f_2 \approx f/2$, we can
thus ensure an overall stabilisation time of $O(f)$, irrespectively of the
number of recursion levels. Formally, we prove the following theorem:

\begin{restatable}{theorem}{twoblocks}\label{thm:twoblocks}
Let $c, n > 1$ and $f <n/3$. Define $n_0=\lfloor n/2 \rfloor$, $n_1 = \lceil
n/2 \rceil$, $f_0 = \lfloor (f-1)/2 \rfloor$, $f_1 = \lceil (f-1)/2 \rceil$, and
$\tau = 3(f+2)$. If for $i \in \{0,1\}$ there exist synchronous counters $\vec
A_i \in \mathcal{A}(n_i, f_i, c_i)$ such that $c_i = 3^i \cdot 2\tau$, then
there exists a synchronous $c$-counter $\vec B \in \mathcal{A}(n,f,c)$ that
 \begin{itemize}[noitemsep]
  \item stabilises in $T(\vec B) = \max \{ T(\vec A_0),T(\vec A_1) \} + O(f)$ rounds, and
  \item has state complexity of $S(\vec B) = \max \{ S(\vec A_0),S(\vec A_1) \} + O(\log f + \log c)$ bits.
 \end{itemize}
\end{restatable}

We fix the notation of this theorem for the remainder of this section.
Moreover, for notational convenience we abbreviate $T = \max \{ T(\vec A_0),T(\vec A_1) \}$ and $S = \max \{ S(\vec A_0),S(\vec A_1) \}$.

\subsection{Agreeing on a Common Counter (Once in a While)\label{ssec:common}}

In this part, we construct a counter that will eventually count consistently
at all nodes for $\tau$ rounds. The $\tau$-counter then will be used as a
common clock for executing the phase king algorithm.

We partition the set of nodes $V = V_0 \cup V_1$ such that $V_0\cap V_1=\emptyset$,
$|V_0| = n_0$ and $|V_1| = n_1$. We refer to the set $V_i$ as
\emph{block $i$}. For each $i \in \{0,1\}$, the nodes in set $V_i$ execute the
algorithm $\vec A_i$. In case block $i$ has more than $f_i$ faults, we call the
block $i$ faulty. Otherwise, we say that block $i$ is correct. By construction,
at least one of the blocks is correct. Hence, there is a correct block $i$ for
which $\vec A_i$ stabilises within $T$ rounds, that is, nodes in block $i$ output a
consistent $c_i$-counter in rounds $r\geq T$.

\begin{lemma}\label{lemma:stable-block}
For some $i \in \{0,1\}$, block $i$ is correct.
\end{lemma}
\begin{proof}
By choice of $f_i$, we have $f=f_0+f_1+1$. Hence, at least one of the sets $V_i$ will contain at most $f_i$ faults.
\end{proof}

Next, we apply the typical threshold voting mechanism employed by most Byzantine
tolerant algorithms in order to filter out differing views of counter values that
are believed to be consistent. This is achieved by broadcasting candidate counter
values and applying a threshold of $n-f$ as a consistency check, which guarantees
that at most one candidate value from the set $[c]$ can remain. In case the threshold
check fails, a fallback value $\bot \notin [c]$ is used to indicate an inconsistency.
This voting scheme is applied for both blocks concurrently, and all
nodes participate in the process, so we can be certain that fewer than one third
of the voters are faulty.

In addition to passing this voting step, we require that the counters also have
behaved consistently over a sufficient number of rounds; this is verified by the
obvious mechanism of testing whether the counter increases by $1$ each round and
counting the number of rounds since the last inconsistency was detected.

In the following, nodes frequently examine a set of values, one broadcast
by each node, and determine majority values. Note that \emph{Byzantine nodes may send
different values to different nodes}, that is, it may happen that correct nodes
output different values from such a vote. We refer to a \emph{strong majority}
as at least $n-f$ nodes supporting the same value, which is then called the
\emph{majority value}. If a node does not see a strong majority, it outputs the
symbol $\bot$ instead.  Clearly, this procedure is well-defined for
$f<n/2$.

We will refer to this procedure as a \emph{majority vote}, and slightly abuse
notation by saying ``majority vote'' when, precisely, we should talk of ``the
output of the majority vote at node $v$''. Since we require that $f<n/3$, the
following standard argument shows that for each vote, there is a unique value
such that each node either outputs this value or $\bot$.
\begin{lemma}\label{lemma:majority}
If $v,w\in V\setminus {\cal F}$ both observe a strong majority, they output the
same majority value.
\end{lemma}
\begin{proof}
Fix any set $A$ of $n-f$ correct nodes. For $v$ and $w$ to observe strong majorities for different values, for each value $A$ must contain $n-2f$ nodes supporting it. However, as correct nodes broadcast the same value to each node, this leads to the contradiction that $|A|\geq 2(n-2f)=n-f+(n-3f)>n-f=|A|$.
\end{proof}

We now put this principle to use. In the following, we will use the notation $x(v,r)$ to refer to the value of local variable $x$ of node $v$ in round $r$. As we consider self-stabilising algorithms, the nodes themselves are not aware of what is the value of $r$. We introduce the following local variables for
each node $v \in V$, block $i \in \{0,1\}$, and round $r>0$ (see Tables \ref{tab:variables} and \ref{tab:variables2}):
\begin{itemize}
 \item $m_i(v,r)$ stores the most frequent counter value in block $i$ in round $r$, which is determined from the broadcasted output variables of ${\vec A}_i$ with ties broken arbitrarily,
 \item $M_i(v,r)$ stores the majority vote on $m_i(v,r-1)$,
 \item $w_i(v,r)$ is a cooldown counter which is reset to $2c_1$ whenever
 the node perceives the counter of block $i$ behaving inconsistently, that is,
 $M_i(v,r)\neq M_i(v,r-1)+1\bmod c_i$. Note that this test will automatically fail
 if either value is $\bot$. Otherwise, if the counter behaves consistently, $w_i(v,r)=\max\{w_i(v,r-1)-1,0\}$.
\end{itemize}
Clearly, these variables can be updated based on the local values from the
previous round and the states broadcasted at the beginning of the current round.
This requires nodes to store $O(\log c_i)=O(\log f)$ bits. 

\begin{table}
\centering
\begin{tabular}{@{}l@{\qquad}l@{\qquad}l@{}}
\toprule
Variable & Range & Description \\
\midrule
$m_i(v,r)$ & $[c_i]$ & the most frequent value observed for the $\vec A_i$ counter of block $i$ \\
$M_i(v,r)$ & $[c_i] \cup \{\bot\}$ & the result of majority vote on $m_i(\cdot,r-1)$ values \\
$w_i(v,r)$ & $[c_1+1]$ & ``cooldown counter'' that is reset if block $i$ behaved inconsistently \\
\midrule
$d_i(v,r)$ & $[c_i] \cup \{\bot\}$ & observation on what seems to be the counter output of block $i$ \\
$\ell_i(v,r)$ & $\{0,1,\bot\}$ & the value of the ``leader pointer'' for block $i$ \\
$\ell(v,r)$ & $\{0,1,\bot\}$ & leader pointer used by node $v$ \\
$d(v,r)$ & $[\tau]$ & once-in-a-while round counter for clocking phase king \\
\midrule
$a(v,r)$ & $[c] \cup \{\infty\}$ & the output of the new $c$-counter we are constructing \\
$b(v,r)$ & $\{0,1\}$ & helper variable for the phase king algorithm \\
\bottomrule
\end{tabular}
\caption{The local state variables used in the boosting construction.}\label{tab:variables}
\end{table}

\begin{table}
\centering
\begin{tabular}{@{}l@{\qquad}l@{\qquad}l@{}}
\toprule
Variable & Block $i$ is correct & Block $i$ is faulty \\
\midrule
$m_i(v,r)$ & consistent counter & arbitrary values \\
$M_i(v,r)$ & consistent counter & $\bot$ or some consistent value \\
$d_i(v,r)$ & consistent counter & $\bot$ or some consistent counter \\
$\ell_i(v,r)$ & consistent pointer & $\bot$ or some consistent pointer \\
\bottomrule
\end{tabular}
\caption{Behaviour of local state variables; pointers switch once in $3^i\tau$ rounds.}\label{tab:variables2}
\end{table}

Furthermore, we define the
following derived variables for each $v \in V$, block $i \in \{0,1\}$, and round $r$ (see Tables \ref{tab:variables} and \ref{tab:variables2}):
\begin{itemize}[noitemsep]
  \item $d_i(v,r)=M_i(v,r)$ if $w_i(v,r)=0$, otherwise $d_i(v,r)=\bot$,
  \item $\ell_i(v,r)=\bigl\lfloor d_i(v,r)/(3^i\tau)\bigr\rfloor$ if
  $d_i(v,r)\neq \bot$, otherwise $\ell_i(v,r) = \bot$,
  \item for $v\in V_i$, $\ell(v,r)=\ell_i(v,r)$ if $\ell_i(v,r)\neq \bot$,
  otherwise $\ell(v,r)=\ell_{1-i}(v,r)$, and
  \item $d(v,r)=d_{\ell(v,r)}(v,r)\bmod \tau$ if $\ell(v,r)\neq \bot$, otherwise
  $d(v,r)=0$.
\end{itemize}
These can be computed locally, without storing or communicating additional
values. The variable $\ell(v,r)$ indicates the block that node $v$ currently
considers leader. Note that some nodes may use $\ell_0(\cdot,r)$ as the leader pointer
while some other nodes may use $\ell_1(\cdot,r)$ as the leader pointer, but this is fine:
\begin{itemize}[noitemsep]
    \item all nodes $v$ that use $\ell(v,r) = \ell_0(v,r)$ observe the same value $\ell_0(\cdot,r) \ne \bot$,
    \item all nodes $w$ that use $\ell(w,r) = \ell_1(w,r)$ observe the same value $\ell_1(\cdot,r) \ne \bot$,
    \item eventually $\ell_0(\cdot,r)$ and $\ell_1(\cdot,r)$ will point to the same correct block for $\tau$ rounds.
\end{itemize}

We now verify that $\ell(v,r)$ indeed has the desired properties. To this end, we
analyse $d_i(v,r)$. We start with a lemma showing that eventually a correct
block's counter will be consistently observed by all correct nodes.
\begin{lemma}\label{lemma:voting_correct}
Suppose block $i\in \{0,1\}$ is correct. Then for all $v,w\in V\setminus {\cal
F}$, and rounds $r\geq R=T+O(f)$ it holds that $d_i(v,r)=d_i(w,r)$ and
 $d_i(v,r)=d_i(v,r-1)+1\bmod c_i$.
\end{lemma}
\begin{proof}
Since block $i$ is correct, algorithm ${\vec A}_i$ stabilises within $T({\vec A}_i)$ rounds.
As $f_i < n_i/3$, we will observe correctly
$m_i(v,r+1)=m_i(v,r)+1\bmod c_i$ for all $r\geq T({\vec A}_i)$.
Consequently, $M_i(v,r+1)=M_i(v,r)+1\bmod c_i$ for all $r\geq T({\vec A}_i)+1$.
Therefore, $w_i(v,r)$ cannot be reset in rounds $r\geq T({\vec A}_i)+2$,
yielding that $w_i(v,r)=0$ for all $r\geq T({\vec A}_i)+2+2c_1=T+O(f)$. 
The claim follows from the definition of variable $d_i(v,r)$.
\end{proof}

The following lemma states that if a correct node $v$ does not detect an error in a
block's counter, then any other correct node $w$ that considers the
block's counter correct \emph{in any of the last $2c_1$ rounds}
has a counter value that agrees with $v$.
\begin{lemma}\label{lemma:voting_general}
Suppose for $i \in \{0,1\}$, $v\in V\setminus {\cal F}$, and $r \geq 2c_1 = O(f)$ it
holds that $d_i(v,r) \neq \bot$. Then for each $w \in V \setminus {\cal F}$ and
each $r' \in \{r-2c_1+1,\ldots,r\}$ either 
\begin{itemize}[noitemsep]
 \item $d_i(w,r') = d_i(v,r) - (r-r') \bmod c_i$, or 
 \item $d_i(w,r')=\bot$.
\end{itemize}
\end{lemma}
\begin{proof}
Suppose $d_i(w,r')\neq \bot$. Thus, $d_i(w,r')=M_i(w,r')\neq \bot$. By
\lemmaref{lemma:majority}, either $M_i(v,r')=\bot$ or $M_i(v,r')=M_i(w,r')$.
However, $M_i(v,r')=\bot$ would imply that $w_i(v,r')=2c_1$ and thus
\begin{equation*}
w_i(v,r)\geq w_i(v,r') + r' - r= 2c_1 + r' - r >0,
\end{equation*}
contradicting the assumption that $d_i(v,r)\neq \bot$. Thus,
$M_i(v,r')=M_i(w,r')=d_i(w,r')$. More generally, we get from $r-r'<2c_1$ and
$w_i(v,r)=0$ that $w_i(v,r'')\neq 2c_1$ for all $r''\in \{r',\ldots,r\}$.
Therefore, we have that $M_i(v,r''+1)=M_i(v,r'')+1\bmod c$ for all $r''\in
\{r',\ldots,r-1\}$, implying 
\begin{equation*}
d_i(v,r)=M_i(v,r)=M_i(v,r')+r-r'=d_i(w,r')+r-r',
\end{equation*}
proving the claim of the lemma.
\end{proof}

The above properties allow us to prove a key lemma:
within $T+O(f)$ rounds, there will be $\tau$ consecutive rounds during which the
variable $\ell(v,r)$ points to the same correct block for all correct nodes.
\begin{lemma}\label{lemma:leader}
Let $R$ be as in \lemmaref{lemma:voting_correct}. There is a round $r\leq
R+O(f)=T+O(f)$ and a correct block $i$ so that for all $v\in V\setminus {\cal
F}$ and $r'\in \{r,\ldots,r+\tau-1\}$ it holds that $\ell(v,r')=i$.
\end{lemma}
\begin{proof}
By \lemmaref{lemma:stable-block}, there exists a correct block $i$. 
Thus by \lemmaref{lemma:voting_correct}, variable $d_i(v,r)$ counts correctly during
rounds $r \geq R$. If there is no round $r\in \{R,\ldots,R+c_i-1\}$ such that
some $v\in V\setminus {\cal F}$ has $\ell_{1-i}(v,r)\neq \bot$, then
$\ell(v,r)=\ell_i(v,r)$ for all such $v$ and $r$ and the claim of the lemma
holds true by the definition of $\ell_i(v,r)$ and the fact that $d_i(v,r)$
counts correctly and consistently.

Hence, assume that $r_{0}\in \{R,\ldots,R+c_i-1\}$ is minimal with the
property that there is some $v\in V\setminus {\cal F}$ so that
$\ell_{1-i}(v,r_{0})\neq \bot$. Therefore, $d_{1-i}(v,r_{0})\neq \bot$
and, by \lemmaref{lemma:voting_general}, this implies for all $w\in V\setminus
{\cal F}$ and all $r\in \{r_{0},\ldots,r_{0}+2c_1-1\}$ that either
$d_{1-i}(w,r)=\bot$ or $d_{1-i}(w,r)=d_{1-i}(v,r_{0})+r-r_{0}$. In other
words, there is a ``virtual counter'' that equals $d_{1-i}(v,r_{0})$ in round
$r_{0}$ so that during rounds $\{r_{0},\ldots,r_{0}+2c_1-1\}$
all $d_{1-i}(\cdot,\cdot)$ variables that are not $\bot$ agree with this counter.

Consequently, it remains to show that both $\ell_i$ and the variable
$\ell_{1-i}$ derived from this virtual counter are equal to $i$ for $\tau$ consecutive
rounds during the interval $I = \{r_{0},\ldots,r_{0}+2c_1-1\}$, as then
$\ell(v,r')=i$ for $v\in V\setminus {\cal F}$ and all such rounds $r'$. 

Clearly, the $c_1$-counter consecutively counts from $0$ to $c_1-1$ at least once during
the interval $I = \{r_{0},\ldots,r_{0}+2c_1-1\}$. Recalling that $c_1=6\tau$, we see
that $\ell_1(v,r)=i$ for all $v\in V\setminus {\cal F}$ with $\ell_1(v,r) \neq
\bot$ for some interval $I_1 \subset I$ of $3\tau$ consecutive rounds.
As $c_0=2\tau$, we have that $\ell_0(v,r)=i$ for all $v\in V\setminus {\cal F}$ with $\ell_0(v,r)\neq \bot$ for $\tau$ consecutive rounds during this subinterval $I_1$. Thus, we have an interval $I_0 = \{r, \ldots, r+\tau-1\} \subseteq I_1$ such that for all $r' \in I_0$ we have $\ell_0(v,r'), \ell_1(v,r') \in \{i, \bot\}$ and $\ell_0(v,r') \neq \bot$ or $\ell_1(v,r') \neq \bot$ yielding $\ell(v,r')=i$ for each correct node. Because $r < r_{0}+2c_1-1< R+3c_1=T+O(f)$, this completes the proof.
\end{proof}

Using the above lemma, we get a counter where all nodes eventually count correctly
and consistently modulo $\tau$ for at least $\tau$ rounds.

\begin{corollary}\label{coro:round-counter}
There is a round $r=T+O(f)$ so that for all $v,w \in V \setminus {\cal F}$ it holds that
\begin{enumerate}
 \item $d(v,r)=d(w,r)$ and 
 \item for all $r'\in \{r+1,\ldots,r+\tau-1\}$ we have $d(v,r')=d(v,r'-1)+1\bmod \tau$.
\end{enumerate}
\end{corollary}
\begin{proof}
By \lemmaref{lemma:leader}, there is a round $r=T+O(f)$ and a correct block $i$
such that for all $v\in V\setminus {\cal F}$ we have $\ell(v,r')=i$ for all
$r'\in \{r,\ldots,r+\tau-1\}$. Moreover, $r$ is sufficiently large to apply
\lemmaref{lemma:voting_correct} to $d_i(v,r')=d(v,r')$ for $r'\in
\{r+1,\ldots,r+\tau-1\}$, yielding the claim.
\end{proof}

\subsection{Reaching Consensus\label{ssec:consensus}}

\corollaryref{coro:round-counter} guarantees that all correct nodes eventually agree on a common counter for $\tau$ rounds, i.e., we have a weak counter. We will now use the weak counter to construct a strong counter.

Our construction uses a non-self-stabilising consensus algorithm. The basic idea is that the weak counter serves as the ``round counter'' for the consensus algorithm. Hence we will reach agreement as soon as the weak counter is counting correctly. The key challenge is to make sure that agreement \emph{persists} even if the counter starts to misbehave. It turns out that a straightforward adaptation of the classic phase king protocol~\cite{berman89consensus} does the job. The algorithm has the following
properties:
\begin{itemize}
 \item the algorithm tolerates $f < n/3$ Byzantine failures, 
 \item the running time of the algorithm is $O(f)$ rounds and it uses $O(\log c)$ bits of state, 
 \item if node $k$ is correct, then agreement is reached if all
 correct nodes execute rounds $3k$, $3k+1$, and $3k+2$ consecutively in this order, 
 \item once agreement is reached, it will persist even if nodes execute \emph{different} rounds in arbitrary order.
\end{itemize}

We now describe the modified phase king algorithm that will yield a $c$-counting algorithm.
Denote by $a(v,r) \in [c] \cup \{ \infty \}$ the output value of the algorithm at round $r$.
Here $\infty$ is used as a ``reset state'' similarly to $\bot$ in the previous section. There is also an auxiliary binary value $b(v, r) \in \{0,1\}$. Define the following short-hand for the increment operation modulo $c$:
\[
x \oplus 1 = \begin{cases}
x + 1 \bmod c & \text{if } x \neq \infty,\\
\infty & \text{if } x = \infty.
\end{cases}
\]

\begin{table}
\centering
\begin{tabular}{@{}l@{\qquad}l@{\ \ }l@{}}
\toprule
Set & \multicolumn{2}{@{}l@{}}{Instructions for round $r>0$} \\
\midrule
$I_{3k}$:
& 0a. & If fewer than $n-f$ nodes sent $a(v,r-1)$, set $a(v,r) = \infty$. \\
& 0b. & Otherwise, $a(v,r) = a(v,r-1) \oplus 1$. \\
\midrule
$I_{3k+1}$:
& 1a. & Let $z_j = |\{u \in V : a(u,r-1) = j\}|$ be the number of $j$ values received. \\
& 1b. & If $z_{a(v,r-1)} \ge n-f$, set $b(v,r) = 1$. Otherwise, set $b(v,r) = 0$. \\
& 1c. & Let $z = \min \{ j : z_j > f\}$. \\
& 1d. & Set $a(v,r) = z \oplus 1$. \\
\midrule
$I_{3k+2}$:
& 2a. & If $a(v,r-1) = \infty$ or $b(v,r-1) = 0$, set $a(v,r) = \min\{c-1, a(k,r-1) \} \oplus 1$. \\
& 2b. & Otherwise, $a(v,r) = a(v,r-1) \oplus 1$. \\
& 2c. & Set $b(v,r) = 1$. \\
\bottomrule
\end{tabular}
\caption{The instruction sets for node $v\in V$ in the phase king protocol.}\label{tab:instr}
\end{table}

For $k \in [f+2]$, we define the instruction sets listed in \tableref{tab:instr}. Recall that in the model of computation that we use in this work, in each round all nodes first broadcast their current state (in particular, the current value of $a$), then they receive the messages, and finally they update their local state. The instruction sets pertain to the final part---how to update the local state variables $a$ and $b$ based on the messages received from the other nodes.

First, we show that if the instruction sets are executed in the right order by all correct nodes for a correct leader node $k \in [f+2]$, then agreement on a counter value is established.

\begin{lemma}\label{lemma:establish_agreement}
Suppose that for some correct node $k \in [f+2]$ and a round $r>2$, all non-faulty nodes execute instruction sets $I_{3k}$, $I_{3k+1}$, and $I_{3k+2}$ in rounds $r-2$, $r-1$, and $r$, respectively. Then $a(v,r)=a(u,r) \neq \infty$ for any two correct nodes $u,v\in V$. Moreover, $b(v,r+1)=1$ at each correct node $v \in V$.
\end{lemma}
\begin{proof}
This is essentially the correctness proof for the phase king algorithm. Without loss of generality, we can assume that the number of faulty nodes is exactly $f$. Since we have $f<n/3$, it is not possible that two correct nodes $u,v \in V \setminus \mathcal{F}$ both satisfy $a(v,r-2) \neq a(u,r-2)$ and $a(v,r-2), a(u,r-2) \in [c]$: otherwise, on round $r-2$, nodes $u$ and $v$ would have observed different majority values contradicting \lemmaref{lemma:majority}. Therefore, there exists some $x \in [c]$ such that $a(v,r-2) \in \{x,\infty\}$ for all $v \in V \setminus \mathcal{F}$. Checking $I_{3k+1}$ we get that $a(v,r-1) \in \{x+1 \bmod c,\infty\}$, as no node can see values other than $x$ or $\infty$ more than $f$ times when executing instruction 1c.

To prove the claim, it remains to consider two cases when executing instructions in $I_{3k+2}$. In the first case, all non-faulty nodes execute instruction 2a on round $r$. Then $a(u,r) = a(v,r) = \min\{c-1, a(k,r-1)\} \oplus 1 \in [c]$ for any $u,v \in V \setminus \mathcal{F}$. 

In the second case, there is some node $v$ not executing instruction 2a. Hence, $a(v,r-1)\neq \infty$ and $b(v,r-1)=1$, implying that $v$ computed $z_{a(v,r-2)} \geq n-f$ on round $r-1$. Consequently, at least $n-2f>f$ correct nodes $u$ satisfy $a(u, r-2)=a(v, r-2)\neq \infty$. We can now infer that $a(u, r-1) = a(v, r-1)=a(v,r-2)+1\bmod c$ for all correct nodes $u$: instruction 1c must evaluate to $a(v, r-1) \in [c]$ at all correct nodes, because we know that no correct node $u$ satisfies that both $a(u,r-2)\neq a(v,r-2)$ and $a(u,r-2)\neq \infty$. This implies that $a(u,r) = a(v,r) \neq \infty$ for all correct nodes $u$, regardless of whether they execute instruction 2a. Trivially, $b(v,r) = 1$ at each correct node $v$ due to instruction 2c.
\end{proof}

Next, we argue that once agreement is established, it persists---it does not matter any more which instruction sets are executed.
\begin{lemma}\label{lemma:maintain_agreement}
Assume that $a(v,r) = x \in [c]$ and $b(v,r) = 1$ for all correct nodes $v$ in some round~$r$. Then $a(v,r+1) = x+1 \bmod c$ and $b(v,r+1)=1$ for all correct nodes $v$.
\end{lemma}
\begin{proof}
Each node observes at least $n-f$ nodes with counter value $x \in [c]$, and hence at most $f$ nodes with some value $y \neq x$. Let $v$ be a correct node and consider all possible instruction sets it may execute. 

First, consider the case where instruction set $I_{3k}$ is executed. In this case, $v$ increments $x$, resulting in $a(v,r+1) = x+1 \bmod c$ and $b(v,r+1)= 1$. Second, executing $I_{3k+1}$, node $v$ evaluates $z_x \ge n-f$ and $z_y \le f$ for all $y \neq x$. Hence it sets $b(v,r+1)=1$ and $a(v,r+1) = x + 1 \bmod c$. Finally, when executing $I_{3k+2}$, node $v$ skips instruction 2a and sets $a(v,r+1) = x+1\bmod c$ and $b(v,r+1) = 1$.
\end{proof}

\subsection{Proof of Theorem~\ref{thm:twoblocks}}

We now have all the building blocks to devise an $f$-resilient $c$-counter running on $n$ nodes. The idea is as follows: first, we use the construction given in \sectionref{ssec:common} to get a weak $\tau$-counter that eventually counts correctly for $\tau=3(f+2)$ rounds. Concurrently, all nodes execute the modified phase king algorithm given in \sectionref{ssec:consensus} which by \lemmaref{lemma:establish_agreement}~and~\lemmaref{lemma:maintain_agreement} guarantees that all nodes will establish and maintain agreement on the output variable for the $c$-counter.

\twoblocks*
\begin{proof}
First, we apply the construction underlying \corollaryref{coro:round-counter}.
Then we have every node $v \in V$ in each round $r$ execute the instructions for round
$d(v,r)$ of the phase king algorithm from \sectionref{ssec:consensus}. It
remains to show that this yields a correct algorithm~$\vec B$ with stabilisation
time $T(\vec B) = T + O(f)$ and state complexity $S(\vec B) = S + O(\log f +
\log c)$, where $T = \max\{ T(\vec A_i) \}$ and $S = \max \{ S(\vec A_i)\}$.

By \corollaryref{coro:round-counter}, there exists a round $r=T+O(f)$ so that the
variables $d(v,r)$ behave as a consistent $\tau$-counter during rounds $\{r,\ldots,r+\tau-1\}$
for all $v \in V \setminus \mathcal{F}$. 
As there are at most $f$ faulty nodes, there exist at least two correct nodes $v \in [f+2]$. 
Since $\tau = 3(f+2)$, then for at least one correct node $k \in [f+2] \setminus \mathcal{F}$, there is a round $r \le r_k \leq r + \tau -3$ such that $d(w,r_k+h) = 3k + h$ for all $w \in V \setminus {\cal F}$ and $h \in \{0,1,2\}$. 
Therefore, by \lemmaref{lemma:establish_agreement}~and~\lemmaref{lemma:maintain_agreement}, the output variables satisfy $a(v,r') = a(w,r') \in [c]$ for all correct nodes and rounds $r ' \ge r_k+3$.
Thus, the algorithm stabilises in $r_v + 3 \leq r + \tau = r + O(f) = T + O(f)$ rounds.

The bound for the state complexity follows from the facts that, at each node, we
need at most $S$ bits to store the state of ${\vec A}_i$ and
$O(\log \tau + \log c)=O(\log f + \log c)$ bits to store the variables
listed in \tableref{tab:variables}.
\end{proof}

\section{Deterministic Counting}\label{sec:recursion}

In this section, we use the construction given in the previous section to obtain algorithms that only need a small number of state bits. Essentially, all that remains is to recursively apply \theoremref{thm:twoblocks}. Each step of the recursion roughly doubles the resilience in an optimal manner: if we start with an optimally resilient algorithm, we get a new algorithm with higher, but still optimal, resilience. Therefore, to get any desired resilience of $f>0$, it suffices to repeat the recursion for $\Theta(\log f)$ many steps. \figureref{fig:recursion} illustrates how we can recursively apply \theoremref{thm:twoblocks}.

We now analyse the correctness, time and state complexity of the resulting algorithms.

\begin{figure}[t]
 \centering
 \includegraphics{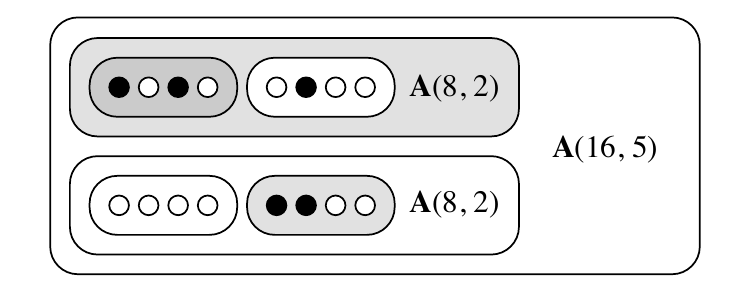}
 \caption{An example on how to recursively construct a 5-resilient algorithm running on 16 nodes. The small circles represent the nodes. Each group of four nodes runs a 1-resilient counter $\vec A(4, 1)$. On top of this, each larger group of 8 nodes runs a 2-resilient counter $\vec A(8, 2)$ attained from the first step of recursion. At the top-most layer, all of the 16 nodes run a 5-resilient counter $\vec A(16,5)$. Faulty nodes are black and faulty blocks are gray.
\label{fig:recursion}}
\end{figure}

\mainthm*
\begin{proof}
We show the claim by induction on $f$. The induction hypothesis is that for all
$f > f' \ge 0$, $c>1$, and $n'>3f'$, we can construct ${\vec B}\in {\cal A}(f',n',c)$
with
\begin{equation*}
  T({\vec B}) \le 1+\alpha f'\sum_{k=0}^{\lceil \log f' \rceil}(1/2)^k \qquad 
  \mbox{and} \qquad S({\vec B}) \le \beta(\log^2 f' +\log c),
\end{equation*}
where $\alpha$ and $\beta$ are sufficiently large constants and for $f'=0$ the sum is
empty, that is, $T({\vec B}) \le 1$. As $\sum_{k=0}^{\infty} (1/2)^k=2$, the time bound will be $O(f')$. 

Note that for $f \ge 0$ it is sufficient to show the claim for
$n(f)=3f+1$, as we can easily generalise to any $n>n(f)$ by running $\vec B$ on
the first $n(f)$ nodes and letting the remaining nodes follow the majority
counter value among the first $n(f)$ nodes executing the algorithm; this increases the
stabilisation time by one round and induces no memory overhead.

For the base case, observe that a $0$-resilient $c$-counter of $n(0)=1$ node is
trivially given by the node having a local counter. It stabilises in $0$ rounds and
requires $\lceil\log c\rceil$ state bits. As pointed out above, this implies a
$0$-resilient $c$-counter for any $n$ with stabilisation time $1$ and $\lceil
\log c\rceil$ bits of state.

For the inductive step to $f$, we apply \theoremref{thm:twoblocks} with the parameters $n_0 = \lfloor n/2 \rfloor$, $n_1 = \lceil n/2 \rceil$, $f_0 = \lfloor(f-1)/2) \rfloor$, $f_1 = \lceil(f-1)/2)\rceil$, $\tau = 3(f+2)$ and $c_i = 3^i\cdot2\tau$. Since $f_i\leq f/2$ and $n_i>3f_i$, for $i \in \{0,1\}$, the induction hypothesis gives us algorithms ${\vec A}_i(n_i,f_i,c_i)$. Now by applying \theoremref{thm:twoblocks} we get an algorithm $\vec B$
with 
\begin{align*}
T(\vec B) &= \max\{ T(\vec A_0), T(\vec A_1) \} + O(f) \\
          &\leq 1+\frac{\alpha f}{2}\sum_{k=0}^{\lceil \log f/2 \rceil}\left(\frac{1}{2}\right)^k+O(f) \\
          &=1+\alpha f \sum_{k=1}^{\lceil \log f \rceil}\left(\frac{1}{2}\right)^k+O(f) \\ 
          &\leq 1+\alpha f \sum_{k=0}^{\lceil \log f \rceil}\left(\frac{1}{2}\right)^k,
\end{align*}
where in the second to last step we use that $\alpha$ is a sufficiently large constant. Since the sum is at most 2, we get that $T(\vec B) = O(f)$. Moreover, the state complexity is bounded by
\begin{align*}
S({\vec B}) &= \max\{ S(\vec A_0), S(\vec A_1) \} + O(\log f + \log c) \\
            &\leq \beta \left(\log^2 \frac{f}{2} + \log \frac{f}{2}\right) + O(\log f + \log c) \\
            &\leq \beta \left(\log^2 f +\log c\right),
\end{align*}
where we exploit that $\beta$ is a sufficiently large constant. Hence, $S(\vec B) = O(\log^2 f + \log c)$, the induction step succeeds, and the proof is complete.
\end{proof}

\section{Reducing the Number of Bits Communicated}\label{sec:silence}

In this section, we discuss how to reduce the number of bits broadcast by a node \emph{after} stabilisation. We consider the following extension of the model of computation: instead of a node always broadcasting its current state, we allow it to broadcast an arbitrary message (including an empty message) each round. Formally, this entails that we extend the definition of an algorithm by (1) introducing a new function $\mu \colon [n] \times X \to \mathcal{M}$ that maps the current state $x$ to a message $\mu(x)$ which is broadcast and (2) modify the state transition function to map the old internal state and the vector of received messages to a new state, that is, the new state transition function has the form $g' : [n] \times X \times \mathcal{M}^n \to X$.

First, we show how to construct counters that only send $O(1 + B \log B)$ bits every $\kappa$ rounds, where $B = O(\log c / \log \kappa)$, while increasing the stabilisation time only by an additive $O(\kappa)$ term, where $\kappa = \Omega(f)$ is a parameter. In particular, we show that for polynomial-sized counters with optimal resilience, the algorithm only needs to communicate an asymptotically optimal number of bits after stabilisation:

\SILENCE*

We start by outlining the high-level idea of the approach,
then give a detailed description of the construction we use, and finally prove 
the main results of this section.

\subsection{High-Level Idea}

The techniques we use are very similar to the ones we
used for deriving \theoremref{thm:main-theorem}. 
Essentially, we devise a ``silencing wrapper'' for algorithms given by \theoremref{thm:main-theorem}. Let $\vec A$ be such a counting algorithm.
The high-level idea and the key ingredients are the following:
\begin{itemize}
  \item The goal is that nodes eventually become \emph{happy}: they assume 
  stabilisation has occured and check for counter consistency only every $\kappa$ rounds (as
  self-stabilising algorithms always need to verify their output).
  \item Happy nodes do not execute the underlying algorithm ${\vec A}$.
  \item Using a cooldown counter with similar effects as shown in
  \lemmaref{lemma:voting_general}, we enforce that all happy nodes output consistent counters.
  \item We override the phase king instruction of ${\vec A}$ if at least
  $n-2f\geq f+1$ nodes claim to be happy and propose a counter value $x$. In that case nodes adjust their counter output to match $x$. If there is
  no strong majority of happy nodes supporting a counter value, either all nodes   become unhappy or all correct nodes reach agreement and start counting correctly.
  \item If all correct nodes are unhappy, they execute ${\vec A}$ ``as is'' reaching agreement eventually.
  \item The counters are used to make all nodes concurrently switch
  their state to being happy, in a
  way that does not interfere with the above stabilisation process.
\end{itemize}

We will show that happy nodes can communicate their counter values
very efficiently in a manner that self-stabilises within $\kappa$ rounds. As
their counter increases by $1$ modulo $c$ every round (or they become unhappy),
they can use multiple rounds to encode a counter value; the recipient simply
counts locally in the meantime.

\subsection{The Silencing Wrapper}

Let $\vec A \in \mathcal{A}(n,f,c)$ be an algorithm given by
\theoremref{thm:main-theorem} and let $c = j\kappa$ for any $j > 0$ and
$\kappa>T(\vec A)$. We use the short-hand $T=T(\vec A)$ throughout this section. Let $a(v,r)$ be the
output of the synchronous counting algorithm for node $v$ in round $r$. Recall that by a
\emph{strong majority} we mean that at least $n-f$ received messages support a value. We now
modify $\vec A$ so that it meets the additional requirement of little
communication after stabilisation.

We introduce two new variables: a cooldown counter $t(v,r) \in [T+1]$ and a
``happiness'' indicator $h(v,r) \in \{0,1\}$. These are updated according to the
following rules in every round $r>0$:
\begin{enumerate}
 \item Set $t(v,r) = T$ if there was no strong majority of nodes $w$ with
 $a(w,r-1)=a(v,r-1)$ or $a(v,r)\neq a(v,r-1)+1\bmod c$.
 Otherwise, decrement the counter, that is, $t(v,r) = \max\{ 0, t(v,r-1) - 1\}$.
 \item Set $h(v,r) = 0$ if $h(v,r-1)=1$, but there was no strong majority of
 nodes $w$ with $h(w,r-1)=1$ and $a(w,r-1)=a(v,r-1)$, or if $t(v,r) > 0$.
 Set $h(v,r)=1$ if $t(v,r-1) = 0$ and $a(v,r-1) = 0 \bmod \kappa$. Otherwise,
 $h(v,r)=h(v,r-1)$.
 \item If $h(v,r) = 0$, execute a single step of $\vec A$ \emph{except} for the phase king instructions given in \tableref{tab:instr}. The counter value $a(v,r+1)$ is updated according to the next rule. 
 \item If received $n-2f$ times a value $a(w,r)=x$ from nodes with
 $h(w,r)=1$, set $a(v,r+1) = x + 1 \bmod c$; if there are two such values $x$,
 it does not matter which is chosen. Otherwise, execute \emph{only} the phase king
 instructions of $\vec A$ given in \tableref{tab:instr} as indicated by the once-in-a-while round counter $d(v,r)$ as usual; in particular, this determines $a(v,r+1)$.
\end{enumerate}	

In the following, we say that a node $v \in V \setminus \mathcal{F}$ with value
$h(v,r)=1$ is \emph{happy} in round $r$ and \emph{unhappy} if $h(v,r)=0$.
Moreover, the counters \emph{converge} in round $r$ if for all $v,w \in V
\setminus \mathcal{F}$, it holds that $a(v,r) = a(w,r)$. The idea is to show
that not only do the counters converge (and then count correctly), but also all
correct nodes become happy. As a happy node that remains happy simply increases
its counter value by $1$ modulo $c$, there is no need to explicitly communicate
this except for verification purposes. It is straightforward to exploit this to
ensure that the algorithm communicates very little (explicitly) once all nodes
are happy; we will discuss this after showing stabilisation of the routine.

\subsection{Proof of Stabilisation}

Let us first establish that if the counters converge, they will keep counting
correctly and correct nodes will become happy within $O(\kappa + T)$ additional
rounds for any parameter $\kappa > T$.

\begin{lemma}\label{lemma:persistence}
If the counters converge in round $r$, then $a(u,r') = a(v,r') =
a(u,r)+(r-r')\bmod c$ for all $u,v\in V\setminus {\cal F}$ and $r' \ge r$.
\end{lemma}
\begin{proof}
Since the counters have converged, there is a strong majority of nodes supporting the same value. Hence, variable $a(u,r')$ is updated according to Rule 4. As all counter values from
correct nodes are identical, it does not matter whether these nodes are happy or
not; either way, the counters are increased by $1$ modulo $c$
(cf.~\lemmaref{lemma:maintain_agreement}).
\end{proof}

\begin{lemma}\label{lemma:happiness}
If the counters converge in round $r$, then for all rounds $r'\geq r+T+\kappa$
and all nodes $v\in V\setminus{\cal F}$ we have $h(v,r')=1$.
\end{lemma}
\begin{proof}
By \lemmaref{lemma:persistence}, the agreement on output values will persist once reached. Hence, at all nodes
$v\in V\setminus{\cal F}$ we have $t(v,r')=0$ in all rounds $r'\geq r+T$ by Rule 1. Therefore, there
is a round $r'\leq r+T+\kappa$ so that $t(v,r')=0$ and $a(v,r')=0 \bmod \kappa$ at all such
$v$. Consequently, all correct nodes jointly set $h(v,r'+1)=1$. By induction on
the round number, we see that no such node sets $h(v,r'')=0$ for $r''>r'+1$, as
there is always a strong majority of $n-f$ happy and correct nodes supporting
the (joint) counter value.
\end{proof}

We now proceed to show that the counters converge within $O(\kappa + T)$ rounds.
The first step is to observe that if no correct node is happy, then algorithm ${\vec A}$ is run without modification, and hence, the counters converge in $T$ rounds.

\begin{lemma}\label{lemma:2}
Let $r \ge T$. If for all $v \in V \setminus \mathcal{F}$ and $r' \in \{r-T+1, \ldots, r\}$, we
have $h(v,r') = 0$, then the counters converge in round $r+1$.
\end{lemma}
\begin{proof}
Since $h(v,r')=0$, each node $v$ applies Rule 3 in any such round $r'$. As there
are no happy nodes in round $r'$, a node can never receive the same counter value
from more than $f$ nodes that (claim to be) happy. Hence, Rule 4 boils down to
just updating $a(v,r')$ according to the rules of~$\vec A$. As $T = T(\vec
A)$, algorithm $\vec A$ stabilises and thus $a(v,r)=a(w,r)$ for all $v,w\in V\setminus {\cal F}$.
\end{proof}

To deal with the case that some nodes may be happy (which entails that not all
nodes may execute ${\vec A}$ correctly, destroying its guarantees), we argue
that ongoing happiness also implies that the counters converge. To this end, we
first show that the cooldown counters $t(v,r)$ ensure that correct nodes whose
counters are $0$ count correctly and agree on their counter values. This is
shown analogously to \lemmaref{lemma:voting_general}.

\begin{lemma}\label{lemma:3}
Let $r > T$ and $v,w \in V \setminus \mathcal{F}$. If $t(v,r)=t(w,r')=0$ for $r'\in
\{r-T+1,\ldots,r\}$, then $a(v,r)=a(w,r')+r-r' \bmod c$.
\end{lemma}
\begin{proof}
Since $t(v,r)=0$, by Rule 1 it holds that $t(v,r') \le r-r' < T$. Hence, both $v$ and $w$ saw a strong majority
of nodes $u$ with $a(u,r'-1)=a(v,r'-1)$ and $a(u,r'-1)=a(w,r'-1)$, respectively.
By \lemmaref{lemma:majority}, it follows that $a(v,r'-1)=a(w,r'-1)$. Likewise,
$t(v,r'')\neq T$ for rounds $r'<r''\leq r$, implying that
$a(v,r)=a(v,r')+r-r'\bmod c$, and $a(w,r')=a(w,r'-1)+1\bmod
c=a(v,r')$.
\end{proof}

Except for the initial rounds, the above lemma implies that happy nodes always have
the same counter value: by Rule 2, a node $v$ with $h(v,r) = 1$ must have $t(v,r) = 0$. A node remaining happy thus entails that \emph{every} node 
receives the same counter value from at least $n-2f \ge f+1$ happy nodes, and no
other counter value with the same property may be perceived. In other words, a
node staying happy implies that the counters converge.

\begin{lemma}\label{lemma:1}
If $h(v,r-1)=h(v,r)=1$ for some $v\in V\setminus{\cal F}$ and $r>3$, then the
counters converge in round $r+1$.
\end{lemma}
\begin{proof}
By Rule 2, any node $w$ with $h(v,r)=1$ satisfies $t(w,r)=0$. We apply
\lemmaref{lemma:3} to see that, for any $w\in V\setminus {\cal F}$ that is happy
in round $r-1$, we have that $a(v,r-1)=a(w,r-1)$. As $h(v,r)=h(v,r-1)=1$, node~$v$
observed a strong majority of happy nodes $w$ with $a(v,r-1)=a(w,r-1)$ in round
$r-1$, implying that all nodes received this counter value from at least
$n-2f\geq f+1$ happy nodes. Together with Rule 4, these observations imply that
$a(u,r)=a(v,r-1)+1\bmod c$ for all $u\in V\setminus{\cal F}$.
\end{proof}

Using these lemmas and the fact that nodes may become happy only after counting
consistently for sufficiently long \emph{and} when their counters are $0$ modulo
$\kappa>T$, we can show that the counters converge in all cases.

\begin{lemma}\label{lemma:convergence}
Within $O(\kappa)$ rounds, the counters converge.
\end{lemma}
\begin{proof}
Either all $v\in V\setminus {\cal F}$ with $h(v,3)=1$ set $h(v,4)=0$ or
\lemmaref{lemma:1} shows the claim. If there are no nodes $v$ with $h(v,r)=1$
for $r\in\{4,\ldots,T+3\}$, then \lemmaref{lemma:2} shows the claim. Hence,
assume that there is some node $v$ with $h(v,r)=1\neq h(v,r-1)$ for some minimal
$r\in \{4,\ldots,T+3\}$. Again, either $h(v,r+1)=0$ for all such nodes or we can
apply \lemmaref{lemma:1}; thus assume the former in the following.

Suppose for contradiction that there is a node $w$ with $h(w,r')=1$ for
a minimal $r'\in \{r+1,\ldots,r+T$\}. As $r'$ is minimal and all nodes with
$h(v,r)=1$ have $h(v,r+1)=0$, it must hold that $h(w,r'-1)=0$. Hence,
$t(w,r'-1)=0=t(v,r-1)$. By \lemmaref{lemma:3}, this implies that
$a(w,r'-1)=a(v,r-1)+r-r' \bmod c$. However, $\kappa>T$, $0<r-r'\leq T$, and
$a(v,r-1)=0\bmod \kappa$, implying that $a(w,r'-1)\neq 0\bmod \kappa$, which
(by Rule 2) is a contradiction to $h(w,r')=1\neq h(w,r'-1)$.

We conclude that $h(v,r')=0$ for all $v$ and $r'\in \{r+1,\ldots,r+T\}$. The
claim follows by applying \lemmaref{lemma:2}.
\end{proof}

We now can conclude that within $O(\kappa)$ rounds, the algorithm stabilises in the sense that all nodes become happy and count correctly and consistently.

\begin{corollary}\label{coro:stabilisation}
There exists a round $R = O(\kappa)$ such that for all $v \in V \setminus \mathcal{F}$ and $r \ge R$, it holds that $h(v,r)=1$, and $a(v,r) = a(v,r-1) + 1 \bmod c$, and $a(v,r) = a(w,r)$ for all $w \in V \setminus \mathcal{F}$.
\end{corollary}
\begin{proof}
By \lemmaref{lemma:convergence} we get that there exists a round $r' = O(\kappa)$ in which the counters converge. Since $r' + T + \kappa = O(\kappa)$, happiness follows from \lemmaref{lemma:happiness} and agreement follows from \lemmaref{lemma:persistence}.
\end{proof}

\subsection{Reducing the Communication Complexity after Stabilisation}

As noted earlier, the counter variables for happy nodes count modulo $c$. Hence, it is
trivial to deduce the counter value of a happy node from its counter value in an
earlier round. Moreover, happy nodes do not execute algorithm ${\vec A}$. Therefore, we
can change the encoding of the happy nodes' counter values to reduce the communication
complexity after stabilisation.

\begin{corollary}\label{coro:communication}
Suppose happy nodes communicate their counter values by any method that stabilises
in $\kappa$ rounds, then the algorithm presented in this section retains its
properties, except that its stabilisation time increases by an additive
$\kappa$ rounds.
\end{corollary}

The above immediately implies that happy nodes $v$ could simply transmit the $a(v,r)$
only in rounds $r$ when $a(v,r)\bmod \kappa = 0$ and perform no other
communication. The fact that $v$ does not transmit readily implies that it is
happy, permitting to derive its counter value by counting from the most recent value
$v$ transmitted. Moreover, by \lemmaref{lemma:3} the output counters of happy nodes 
agree after $O(1)$ rounds. Thus, a single local counter suffices for verification 
yielding a cost of using only $\lceil \log c \rceil$ additional bits of memory per node.

Clearly, this trivial encoding mechanism stabilises in $\kappa$ rounds. However,
we can do much better. For simplicity, we do not try to give a tight bound here.

\begin{lemma}\label{lemma:bits}
    Happy nodes can communicate their counter values by sending only $O(1+B \log B)$ bits per $\kappa$ rounds, where $B = O(\log c / \log \kappa)$, in a way that stabilises in $\kappa$ rounds.
\end{lemma}
\begin{proof}
First, we fix two unique bit strings \textsc{happy} and \textsc{unhappy} both having 
a length of $O(1)$ bits. 
We mark all messages from unhappy nodes with the header \textsc{unhappy}.
Happy nodes $v\in V\setminus{\cal F}$ send the bit string \textsc{happy} in
rounds $r$ when $a(v,r) \bmod \kappa = 0$. In this and the subsequent $\kappa-1$
rounds, they furthermore send up to $b$ bits in order to encode the value of $a(v,r) \in [c]$, where they avoid the two excluded unique bit strings \textsc{happy} and \textsc{unhappy}. 
Since we are only interested in the asymptotic behaviour, we may neglect these possible collisions
and determine how large $b$ must be so that in $\kappa$ rounds we can encode $c$
different values.

    Since there are $\kappa$ rounds in which to broadcast a message, we can think each round as being a bin containing the bits broadcast by a node. Suppose we have $B = b/\log b$ uniquely labelled balls that we can place in $\kappa$ different bins. This way we can encode $B$-length strings over an alphabet of size $\kappa$ by interpreting each ball in a bin $i \in [\kappa]$ as giving the indices for the symbol $i$. This allows us to encode a total of $\kappa^B$ distinct values.
    
    Since encoding the unique label of a single ball takes $O(\log B)$ bits and we can use constant-sized delimiters when encoding the set of balls in a single bin, we need $O(B \log B)$ bits to encode all the values. Thus, each node communicates a total of $O(B \log B) = O(b)$ bits during the course of $\kappa$ rounds. In order to encode  $c$ different values, it suffices to satisfy $c \le \kappa^B$. This can be done by choosing $B \ge \log c / \log \kappa$. Taking into account the bits for delimiters and the \textsc{happy} string, the claim follows.
\end{proof}

Overall, we obtain the following theorem.
\begin{theorem}\label{thm:SILENCE!}
For any integers $n > 1$, $f < n/3$, $\kappa = \Omega(f)$, and $c = \kappa j$ for $j>0$, there exists an $f$-resilient synchronous $c$-counter that runs on $n$ nodes,
stabilises in $O(\kappa)$ rounds, and requires $O(\log^2 f+\log c)$ bits to
    encode the state of a node. Moreover, once stabilised, nodes send only $O(1+B\log B)$ bits per $\kappa$ rounds, where $B = O(\log c / \log \kappa)$.
\end{theorem}
\begin{proof}
Let $\vec A \in \mathcal{A}(n,f,c)$ be an algorithm given by \theoremref{thm:main-theorem}. As $T(\vec A) = \Theta(f)$, for any $\kappa > T(\vec A)$, the claim now directly follows from Corollaries \ref{coro:stabilisation} and \ref{coro:communication} and \lemmaref{lemma:bits}, where we note that only
a constant number of variables of size at most $\max\{T({\vec A}),c\}$ need to
be encoded in the state of a node.
\end{proof}
We remark that since $\kappa>T({\vec A})=\Theta(f)$, in case of optimal
resilience and $c = n^{O(1)}$, it holds that $B = O(1)$, and thus also, $O(1+B \log B) =O(1)$.

\SILENCE*
\begin{proof}
All properties except for the optimality of the last point follow from the
choice of parameters by picking $\kappa=\Theta(n)$ in \theoremref{thm:SILENCE!}. The
claimed optimality follows from the fact that in order to prove to a node that
its counter value is inconsistent with that of others, it must receive messages
from at least $f+1 = \Theta(n)$ nodes; to guarantee stabilisation in $O(n)$
rounds, this must happen every $\Omega(n)$ rounds for each correct node.
\end{proof}

\section{Sending Fewer Messages}\label{sec:random}

So far we have considered the size of messages nodes need to broadcast every round. In the case of the algorithm given in \theoremref{thm:main-theorem}, every node will send $S = O(\log^2 f + \log c)$ bits in each round. As there are $\Theta(n^2)$ communication links, the total number of communicated bits in each round is $\Theta(S \cdot n^2)$. In this section, we consider a randomised variant of the algorithm that achieves better message and bit complexities in a slightly different communication model. 

\subsection{Pulling Model}

Throughout this section we consider the following variant of our communication model, where in every synchronous round $t$ each correct node $v$:
\begin{enumerate}[noitemsep]
    \item contacts a subset $C(v,t) \subseteq V$ of other nodes to \emph{pull} information from,
    \item pulls a response message $r_u \in \mathcal{M}$ from every contacted node $u \in C(v,t)$,
    \item updates its local state according to its current state and the responses it received.
\end{enumerate}
Thus, every round $t$ node $v$ obtains a message vector $\vec m = \langle m_0, \ldots m_{n-1} \rangle$, where $m_u = r_u$ if $u \in C(v,t)$ and $m_u = \bot$, otherwise. Besides this modification, the model of computation is as before: node $v$ updates its state using the state transition function $g \colon [n] \times X \times \mathcal{M}^n \to X$ and a correct node $u$ in state $x_u$ responds with the message $\mu(x_u)$, where $\mu \colon X \to \mathcal{M}$ maps the internal state of a node to a message. However in the pulling model, the algorithm also needs to specify the set $C(v,t)$ of nodes it contacts every round. We assume that every correct node chooses this set randomly independent of its internal state. 

As before, faulty nodes may respond with arbitrary messages that can be different for different pulling nodes. We define the (per-node) message and bit complexities of the algorithm as the maximum number of messages and bits, respectively, pulled by a non-faulty node in any round.

This model is motivated by the challenges of designing energy-limited fault-tolerant circuits. We suggest the approach in which each node that makes a request for data also has to provide the energy resources for processing and answering the request. This way by limiting the energy supply of each individual node, we can also effectively limit the total amount of energy wasted due to the actions of the Byzantine nodes. However, to make this approach feasible, we have to design an algorithm in which each non-faulty node needs to make only a few requests for data. In this section we design a randomised algorithm that satisfies this property.

\subsection{High-Level Idea of the Probabilistic Construction}

To keep the number of pulls, and thus number of messages sent, small, we modify the construction of \theoremref{thm:twoblocks} to use random sampling where useful. Essentially, the idea is to show that \emph{with high probability} a small set of sampled messages accurately represents the current state of the system and the randomised algorithm will behave as the deterministic one. There are two steps where the nodes rely on information broadcast by the all the nodes: the majority voting scheme over the blocks and the variant of the phase king algorithm. In the following, both are shown to work under the sampling scheme with high probability by using concentration bound arguments. 

More specifically, here \emph{with high probablity} means that for any constant $k \ge 1$ the probability of failure is bounded above by $\totaln^{-k}$ when sampling $K = \Theta(\log \eta)$ messages (where the constants in the asymptotic notation may depend on $k$); here $\eta$ denotes the total number of nodes in the system after the recursive application of the resilience boosting procedure described in \sectionref{sec:recursion}. The idea is to use a union bound over all levels of recursion, nodes, and considered rounds, to show that the sampling succeeds with high probability in all cases. For the randomised variant of \theoremref{thm:main-theorem}, we will require the following additional constraint: when constructing a counter on $n$ nodes, the total number of failures is bounded by $f < \frac{n}{3+\gamma}$, where $\gamma>0$ is constant. 

This allows us to construct \emph{probabilistic synchronous $c$-counters} in the sense that we say that the counter stabilises in time $T$, if for each round $t \ge T$ all non-faulty nodes count correctly with probability $1-\eta^{-k}$.

\subsection{Sampling Communication Channels}\label{ssec:sampling}

As discussed, there are two steps in the construction of \theoremref{thm:twoblocks} where we rely on broadcasting: (1) the majority voting scheme for electing a leader block and counter, and (2) the execution of the phase king protocol. For the sake of clarity, we only focus on modifying the basic algorithm, where the nodes broadcast their entire state each round. We start with a sampling lemma we use for both steps. First, recall the following concentration bound for the sum of independent random binary variables:

\begin{lemma}[Chernoff's bound]
Let $X = \sum X_i$ be a sum of independent random variables $X_i \in \{0,1\}$. Then for $0 < \delta < 1$,\[
\Pr[X \le (1-\delta)\E[X]] \le \exp\left(-\frac{\delta^2}{2} \E[X] \right).
\]
\end{lemma}

\begin{lemma}\label{lemma:whp}
    Let $U \subseteq V$ be a non-empty set of nodes such that the fraction of faulty nodes in $U$ is strictly less than $1/(3+\gamma)$. Suppose we sample $K$ nodes $v_0, \ldots, v_{K-1}$ uniformly at random from the set $U$. For a given local variable $x(\cdot,r)$ encoded in the nodes' local state on round $r \ge 0$ and a value $y$, define the random variable
\[
    X_i = \begin{cases}
        1 & \text{if } x(v_i, r) = y \textrm{ and } v_i \notin \mathcal{F}, \\
        0 & \text{otherwise}
    \end{cases}
\]
for each $i \in [K]$ and let $X = \sum_{i=0}^{K-1} X_i$ be the number of $y$ values sampled from \emph{correct} nodes. There exists $K_0(\eta,k,\gamma) = \Theta(\log \eta)$ such that $K \ge K_0$ implies the following with high probability:
 \begin{enumerate}[label=(\alph*)]
     \item If $x(u,r) = y$ for all $u \in U \setminus \mathcal{F}$, then $X \ge 2K/3$.
     \item If a majority of nodes $u \in U \setminus \mathcal{F}$ have $x(u,r)=y$, then $X \ge K/3$.
     \item If $X \ge 2K/3$, then $|\{ x(u,r) = y : u \in U \setminus \mathcal{F} \}| \ge |U\setminus \mathcal{F}|/2$.
 \end{enumerate}
 \label{lemma:sampling}
\end{lemma}
\begin{proof}
    Define $\delta = 1-\frac{2}{3}\cdot\frac{3+\gamma}{2+\gamma}$ and let $\rho < 1/(3+\gamma)$ be the fraction of faulty nodes in $U$.

(a) If all correct nodes $u \in U \setminus \mathcal{F}$ agree on value $x(u,r) = y$, then 
\[
\E[X] = \left(1 - \rho \right)K > \frac{2+\gamma}{3+\gamma}K.
\]
As $\delta$ satisfies $(1-\delta)\E[X] > 2K/3$, it follows from Chernoff's bound that
\begin{align*}
 \Pr\left[X < \frac{2K}{3}\right] &\le \Pr[X < (1-\delta) \E[X]] \\
               &\le \exp\left( -\frac{\delta^2}{2} \E[X] \right) \\
               &\le \exp\left( -\delta^2 \frac{2+\gamma}{2(3+\gamma)} K \right).
\end{align*}
If $K_0(\eta,k,\gamma) = \Theta(\log \eta)$ is sufficiently large, $K\geq K_0(\eta,k,\gamma)$ implies that this probability is bounded by $\eta^{-k}$.

 (b) If a majority of non-faulty nodes $u$ have value $x(u,r)=y$, then 
\[
 \E[X] \ge \frac{1}{2}(1-\rho)K > \frac{1}{2} \cdot \frac{2+\gamma}{3+\gamma}K.
\]
As above, by picking the right constants and using concentration bounds, we get that
\begin{align*}
 \Pr\left[X \le \frac{K}{3}\right] &\le \Pr[X < (1-\delta) \E[X]] \\
                &\le \exp\left( -\frac{\delta^2}{2} \E[X] \right) \\
               &\le \exp\left( -\delta^2 \frac{2+\gamma}{4(3+\gamma)} K_0 \right) \le \eta^{-k}.
\end{align*}

 (c) Suppose the majority of correct nodes have values different from $y$. Define 
\[
    \bar{X_i} = \begin{cases}
        1 & \text{if } x(v_i, r) \neq y \textrm{ and } v_i \notin \mathcal{F}, \\
        0 & \text{otherwise}.
    \end{cases}
\]
and $\bar{X} = \sum_{i=0}^{K-1} \bar{X_i}$ as the random variable counting the number of samples with values different from $y$ and arguing as for (b), we see that
 \begin{align*}
 \Pr\left[X \geq \frac{2K}{3}\right] &=  \Pr\left[\bar{X} < \frac{K}{3}\right] \le \eta^{-k},
 \end{align*}
 where again we assume that $K_0(\eta,k,\gamma) = \Theta(\log \eta)$ is sufficiently large. Thus, $X\geq 2K/3$ implies with high probability that the majority of correct nodes have value~$y$.
\end{proof}

\paragraph{Randomised Majority Voting.}

Recall that in the majority voting scheme, there are four local variables, two for each $i \in \{0,1\}$, whose values depend directly on the messages broadcast by all nodes:
\begin{itemize}
 \item $m_i(v,r)$ stores the most frequent counter value in block $i$ in round $r$, which is determined from the broadcasted output variables of ${\vec A}_i$ with ties broken arbitrarily, and
 \item $M_i(v,r)$ stores the majority vote on $m_i(v,r-1)$.
\end{itemize}

Throughout the remainder of this section, we let $K = \Theta(\log \eta)$ such that $K \ge K_0$ as given by \lemmaref{lemma:sampling}. Let ${m^*_i}(v,r)$ be the sampled version of $m_i(v,r)$; here the value is determined by taking a random sample of size $K$ from the set $V_i$. Analogously, the variable ${M^*_i}(v,r)$ is determined by taking a random sample of size $K$ from the set $V$ and taking the value that appears at least $2K/3$ times in the sample.

\begin{remark} \label{remark:faulty-bound}
 It holds that $f_i/n_i < 1/(3+\gamma)$ for $i \in \{0,1\}$.
\end{remark}

\begin{lemma}\label{lemma:whp-voting}
 Suppose block $i \in \{0,1\}$ is correct. Then for all $v \in V \setminus \mathcal{F}$ and $r \ge T(\vec A_i)$, we have 
 \begin{align*}
  {m^*_i}(v,r) &= m_i(v,r) \\
  {M^*_i}(v,r+1) &= M_i(v,r+1)
 \end{align*}
 with high probability.
\end{lemma}
\begin{proof}
    To show the claim, we will apply \lemmaref{lemma:sampling} with $U = V$ and $U = V_i$. Before this, note that the fraction of faulty nodes in both $V$ and $V_i$ is less than $1/(3+\gamma)$: by assumption we have $f/n < 1/(3+\gamma)$ and by \remarkref{remark:faulty-bound} yields $f_i/n_i < 1/(3+\gamma)$. Thus, in both cases, we satisfy the first condition of \lemmaref{lemma:sampling}.
    
    For the claim regarding variable $m_i$, we apply \lemmaref{lemma:sampling} with $U = V_i$, that is, sample the subset $V_i \subseteq V$ consisting of nodes in block~$i$. Since $|V_i| = n_i$ and $i$ is a correct block, the set $V_i$ contains at most $f_i$ faulty nodes and all correct nodes output the same value $y \in [c_i]$, as $\vec A_i$ has stabilised by round $r \ge T(\vec A_i)$. Moreover, $f_i/n_i < \frac{1}{3+\gamma}$ by \remarkref{remark:faulty-bound}, so statement (a) of \lemmaref{lemma:sampling} yields that with high probability at least a fraction of $2/3$ of the sampled nodes output $y$.

To show the claim for variable $M^*_i$, note that by the previous case, ${m^*_i}(v,r) = m_i(v,r)$ holds for all correct nodes $v$ with high probability. Applying Statement (a) of \lemmaref{lemma:sampling} to the set $V$ and variable ${m^*_i}(v,r)$, we get that at least a fraction of $2/3$ of the samples have the same value.
\end{proof} 

From \lemmaref{lemma:whp-voting} it follows that we get probabilistic---in the sense that the claims hold with high probability---variants of \lemmaref{lemma:voting_correct}, \lemmaref{lemma:voting_general}, and \lemmaref{lemma:leader}. These, in turn, yield the following probabilistic variant of \corollaryref{coro:round-counter}.

\begin{corollary}\label{coro:probabilistic-round-counter}
 There is a round $r = T + O(f)$ so that for all $v,w \in V \setminus \mathcal{F}$ with high probability it holds that
\begin{enumerate}
 \item $d(v,r) = d(w,r)$ and
 \item for all $r' \in \{ r+1, \ldots, r + \tau - 1 \}$ we have $d(v,r') = d(v,r'-1) + 1 \mod \tau$.
\end{enumerate}
\end{corollary}

\paragraph{Randomised Phase King.} To obtain a randomised variant of the phase king algorithm, we modify the threshold votes used in the algorithm as follows. Instead of checking whether at least $n-f$ of all messages have the same value, we check whether at least a fraction of $2/3$ of the sampled messages have the same value. Similarly, when checking for at least $f+1$ values, we check whether a fraction of $1/3$ of the sampled messages have this value.

As a corollary, we get that when using the sampling scheme in the pulling model, the execution of the phase king essentially behaves as in the deterministic broadcast model.

\begin{corollary}\label{coro:randomised_phase_king}
When executing the randomised variant of the phase king protocol from \sectionref{sec:construction} for $\eta^{O(1)}$ rounds, the statements of \lemmaref{lemma:establish_agreement} and \lemmaref{lemma:maintain_agreement} hold with high probability.
\end{corollary}
\begin{proof}
The modified phase king algorithm given in \sectionref{ssec:consensus} uses two thresholds, $n-f$ and $f+1$. As discussed, these are replaced with threshold values of $2K/3$ and $K/3$ when taking $K \ge K_0(\eta, k, \gamma)$ samples. Using the statements of \lemmaref{lemma:whp}, we can argue analogously to the proofs of \lemmaref{lemma:establish_agreement} and \lemmaref{lemma:maintain_agreement}. 

First, to see that \lemmaref{lemma:establish_agreement} holds with high probability, note that from statements (b) and (c) of \lemmaref{lemma:whp}, it follows that if a node samples $2K/3$ times value $y$, then w.h.p.\ other nodes sample at least $K/3$ times the same value (that is, we get the probabilistic version of \lemmaref{lemma:majority}). Now we can follow the same reasoning as in \lemmaref{lemma:establish_agreement}. 

Similarly, it is straightforward to check that \lemmaref{lemma:maintain_agreement} holds with high probability: if all correct nodes agree on $a(\cdot)$, then all correct nodes sample at least $2K/3$ times the same value w.h.p.\ by statement (a) of \lemmaref{lemma:whp}. Thus, analogously as in the proof of \lemmaref{lemma:maintain_agreement}, we get that the agreement persists when executing $I_{3k}$, $I_{3k+1}$, or $I_{3k+2}$ with high probability.

Finally, we can apply the union bound over all $\eta^{O(1)}$ rounds and samples taken by correct nodes ($n-f \le \eta$ per round), that is, in total over $\eta^{O(1)}$ events. By choosing large enough $k = O(1)$, we get that the claim holds with probability $1 - \eta^{-k}$.
\end{proof}

\subsection{Randomised Resilience Boosting}\label{ssec:randomised-layering}

It remains to formulate the probabilistic variant of \theoremref{thm:twoblocks}. To this end, define $\mathcal{P}(n, f, c, \eta, k)$ as the family of probabilistic synchronous $c$-counters on $n$ nodes of resilience $f$. Here, probabilistic means that an algorithm $\vec P \in \mathcal{P}(n, f, c, \eta, k)$ with stabilisation time $T(\vec P)$ merely guarantees that it counts correctly with probability $1-\eta^{-k}$ in any given round $t \ge T(\vec P)$. 

Let $P(\vec P)$ denote the number of messages pulled \emph{per node} by a probabilistic counter $\vec P \in \mathcal{P}(n,f,c, \eta, k)$. For any deterministic algorithm $\vec A \in \mathcal{A}(n,f,c)$, we define $P(\vec A) = n$.

\begin{theorem}\label{thm:randomised-blocks}
Let $c,n > 1$ and $f < n/(3+\gamma)$, where $\gamma > 0$ and $n \le \eta$. Define $n_0 = \lfloor n/2 \rfloor$, $n_1 = \lceil n/2 \rceil$, $f_0 = \lfloor (f-1)/2 \rfloor$, $f_1 = \lceil (f-1)/2 \rceil$ and
$\tau = 3(f+2)$. If for $i \in \{0,1\}$ there exist synchronous counters $\vec
A_i \in \mathcal{A}(n_i, f_i, c_i)$ such that $c_i = 3^i \cdot 2\tau$, then for any sufficiently large $k = O(1)$, there exists a probabilistic synchronous $c$-counter $\vec B \in \mathcal{P}(n,f,c,\eta, k)$ that
 \begin{itemize}[noitemsep]
  \item stabilises in $T(\vec B) = \max \{ T(\vec A_0),T(\vec A_1) \} + O(f)$ rounds, 
  \item has state complexity of $S(\vec B) = \max \{ S(\vec A_0),S(\vec A_1) \} + O(\log f + \log c)$ bits, and
  \item each node pulls at most $P(\vec B) = \max \{ P(\vec A_0), P(\vec A_1) \} + O(\log \eta)$ messages per round.
 \end{itemize}
\end{theorem}
\begin{proof}
 The proof proceeds analogously to the proof of \theoremref{thm:twoblocks}. First, we apply \corollaryref{coro:probabilistic-round-counter} to get a round counter that works once in a while with high probability. We can then use this to clock the randomised phase king and \corollaryref{coro:randomised_phase_king} implies that the new output counter will reach agreement in $O(f)$ rounds with high probability. The time and state complexities are as in the proof of \theoremref{thm:twoblocks}.

To analyse the number of pulls, observe that in \lemmaref{lemma:whp-voting} each node samples twice $K = O(\log \eta)$ messages (from both $V_0$ and $V_1$) and \corollaryref{coro:randomised_phase_king} samples $O(\log \eta)$ messages from all the nodes. Thus, in total, a node $v \in V_i$ samples $O(\log \eta)$ messages in addition to the messages pulled when executing $\vec A_i$.
\end{proof}

Note that we can choose to replace $\vec{A}\in \mathcal{A}(n, f, c)$ by $\vec{Q} \in \mathcal{P}(n, f, c, \eta, k)$ when applying this theorem, arguing that with high probability it \emph{behaves} like a corresponding algorithm $\vec{A}\in \mathcal{A}(n, f, c)$ for polynomially many rounds. Furthermore, note that it is also possible to boost the probability of success, and thus the period of stability, by simply increasing the sample size. For instance, sampling $\polylog \eta$ messages yields an error probability of $\eta^{-\polylog \eta}$ in each round, whereas in the extreme case, by ``sampling'' all nodes the algorithm reduces to the deterministic case.

Using \theoremref{thm:randomised-blocks} recursively as in \sectionref{sec:recursion} for $O(\log f)$ steps, we get the following result.

\begin{theorem}\label{thm:random-alg}
For any integers $c,n > 1$, $f < n/(3+\gamma)$, there exists an $f$-resilient probabilistic synchronous $c$-counter that runs on $n$ nodes, requires $O(\log^2 f+\log c)$ bits to encode the state of a node, has each node pull $O(\log f \log n)$ messages per round, and stabilises in $O(f)$ rounds with probability $1 - n^{-k}$, where $k>0$ is a freely chosen constant.
\end{theorem}

\subsection{Oblivious Adversary}

Finally, we remark that under an \emph{oblivious adversary}, that is, an adversary that picks the set of faulty nodes independently of the randomness used by the non-faulty nodes, we get \emph{pseudorandom} synchronous counters satisfying the following: (1) the execution stabilises with high probability and (2) if the execution stabilises, then all non-faulty nodes will deterministically count correctly. Put otherwise, we can fix the random bits used by the nodes to sample the communication links \emph{once}, and with high probability we sample sufficiently many communication links to non-faulty nodes for the algorithm to (deterministically) stabilise. This gives us the following result.

\begin{corollary}
For any integers $c,n > 1$, $f < n/(3+\gamma)$, there exists a pseudorandom synchronous $c$-counter with resilience $f$ against an oblivious fault pattern that runs on $n$ nodes, requires $O(\log^2 f + \log c)$ bits to encode the state of a node, has each node pull $O(\log f \log n)$ messages per round, and stabilises in $O(f)$ rounds.
\end{corollary}

\section{Conclusions}

In this work, we showed that there exist algorithms for synchronous counting that (1) are deterministic, (2) tolerate the optimal number of faults, (3) have asymptotically optimal stabilisation time, and (4) need to store \emph{and} communicate a very small number of bits between consecutive rounds---something no prior algorithms have been able to do. 

In addition, we discussed two complementary approaches on how to further reduce the total number of communicated bits in the network. The first one is a deterministic construction that lets the nodes communicate only few bits after stabilisation, in order to verify that stabilisation has occurred and that the counters agree. The construction retains all properties (1)--(4), and in particular, when constructing polynomially-sized counters with linear resilience, the algorithm communicates an asymptotically optimal number of bits after stabilisation.

The second technique for reducing the amount of communication is based on random sampling of communication channels. Here, we employed randomisation so that each node needs to communicate only with $\polylog n$ instead of $n-1$ other nodes in the system, thus reducing the number of messages sent from $\Theta(n^2)$ to $\Theta(n \polylog n)$. The trade-off here is that the resulting algorithm has \emph{slightly} suboptimal resilience of $f < n/(3+\gamma)$, where $\gamma > 0$ is a constant, and is merely guaranteed to work for polynomially many rounds with high probability before a new stabilisation phase is required. The latter issue disappears when employing pseudorandomness. In this case, one may simply fix a random topology and the algorithm will not fail again after stabilisation; naturally, this necessitates that the Byzantine faulty nodes are chosen in an oblivious manner, i.e., independently of the topology.

We can also combine both techniques to attain probabilistic counters that during stabilisation communicate $\Theta(n \polylog n)$ bits each round and after stabilisation asymptotically optimal $O(1)$ bits every $\Theta(n)$ rounds.

To conclude the paper, we now wish to highlight some interesting problems that still remain open:
\begin{enumerate}[label=Q\arabic*.]
    \item Our solutions are not adaptive (as defined in~\cite{kutten05adaptive}), as their stabilisation time is not bounded by a function of the number of \emph{actual} permanent faults. Can this be achieved?
    \item Are there algorithms that satisfy (1)--(3), but need to store and communicate substantially fewer than $\log^2 f$ bits? This question has been partially answered in follow-up work~\cite{lenzen16near-optimal}, showing that $O(\log f)$ bits suffice. However, no non-trivial lower bound is known, so it remains open whether $o(\log f)$ bits suffice.
    \item Can the ideas presented in this paper be applied to \emph{randomised} consensus routines in order to achieve sublinear stabilisation time with high resilience and small communication overhead? Again, a partial answer is provied in~\cite{lenzen16near-optimal}: this is possible, but the given solutions may still fail \emph{after} stabilisation (with a very small probability per round). The question thus remains open w.r.t.\ the original problem definition, which requires that after stabilisation the algorithm keeps counting correctly.
\end{enumerate}

Finally, we point out that the recursive approach we employ in this paper can be interpreted as an extension of its similar use in synchronous consensus routines~\cite{berman89consensus,berman92optimal}, where the shared round counter is implicitly given by the synchronous start.
\begin{enumerate}[resume*]
    \item Can a similar recursive approach also be used for deriving improved \emph{pulse synchronisation}~\cite{dolev04clock-synchronization,dolev07actions} algorithms?
\end{enumerate}
Interestingly, no reduction from consensus to pulse synchronisation is known, so there is still hope for efficient deterministic pulse synchronisation algorithms that stabilise in sublinear time.

\paragraph{Acknowledgements.} We thank all the anonymous reviewers for helpful comments. 

\DeclareUrlCommand{\Doi}{\urlstyle{same}}
\renewcommand{\doi}[1]{\href{http://dx.doi.org/#1}{\footnotesize\sf
    doi:\Doi{#1}}} \bibliographystyle{plainnat}
\bibliography{optimal-counting}

\end{document}